\documentclass[english]{article}
\usepackage[T1]{fontenc}
\usepackage[latin9]{inputenc}
\usepackage{babel}
\usepackage{array}
\usepackage{float}
\usepackage{multirow}
\usepackage{amsmath}
\usepackage{amsthm}
\usepackage{amssymb}
\usepackage{stmaryrd}
\usepackage{graphicx}
\usepackage[unicode=true,
 bookmarks=true,bookmarksnumbered=false,bookmarksopen=false,
 breaklinks=false,pdfborder={0 0 0},pdfborderstyle={},backref=false,colorlinks=false]
 {hyperref}
\hypersetup{pdftitle={A full complexity dichotomy for immanants},
 pdfauthor={Radu Curticapean}}

\makeatletter

\providecommand{\tabularnewline}{\\}

\theoremstyle{definition}
\newtheorem*{example*}{\protect\examplename}
\theoremstyle{plain}
\newtheorem{thm}{\protect\theoremname}
\theoremstyle{definition}
\newtheorem{defn}[thm]{\protect\definitionname}
\theoremstyle{remark}
\newtheorem{rem}[thm]{\protect\remarkname}
\theoremstyle{plain}
\newtheorem{lem}[thm]{\protect\lemmaname}
\theoremstyle{plain}
\newtheorem{fact}[thm]{\protect\factname}
\theoremstyle{plain}
\newtheorem{cor}[thm]{\protect\corollaryname}
\theoremstyle{remark}
\newtheorem*{rem*}{\protect\remarkname}

\usepackage{ytableau}

\usepackage{graphicx}
\usepackage{xcolor}
\usepackage{transparent}

\usepackage{microtype}

\definecolor{cyan1}{HTML}{33c5ff}
\definecolor{cyan2}{HTML}{33ffff}
\definecolor{cyan3}{HTML}{33ff74}

\usepackage{fullpage}

\makeatother

\providecommand{\corollaryname}{Corollary}
\providecommand{\definitionname}{Definition}
\providecommand{\examplename}{Example}
\providecommand{\factname}{Fact}
\providecommand{\lemmaname}{Lemma}
\providecommand{\remarkname}{Remark}
\providecommand{\theoremname}{Theorem}

\begin{document}
\title{A full complexity dichotomy for immanant families}
\author{Radu Curticapean\thanks{
IT University of Copenhagen, Basic Algorithms Research Copenhagen.  Supported by VILLUM Foundation grant 16582. racu@itu.dk}}

\maketitle
\global\long\def\weight{\mathrm{wt}}%
\global\long\def\height{\mathrm{ht}}%
\global\long\def\bpar{b}%
\global\long\def\imbpar{\iota}%

\global\long\def\per{\mathrm{per}}%
\global\long\def\imm{\mathrm{imm}}%
\global\long\def\OnionCyc{\mathrm{\#OnionCycles}}%
\global\long\def\CC{\mathrm{\#CC}}%
\global\long\def\PerfMatch{\mathrm{\#PerfMatch}}%
\global\long\def\match{\mathrm{\#Match}}%
\global\long\def\immProb{\mathrm{Imm}}%

\global\long\def\tablLR{\mathcal{L}}%

\global\long\def\coeff#1#2{[#1]\,#2}%

\global\long\def\FPT{\mathsf{FPT}}%
\global\long\def\sharpETH{\mathsf{\#ETH}}%
\global\long\def\sharpP{\mathsf{\#P}}%
\global\long\def\VP{\mathsf{VP}}%
\global\long\def\VW{\mathsf{VW[1]}}%
\global\long\def\VFPT{\mathsf{VFPT}}%
\global\long\def\VNP{\mathsf{VNP}}%
\global\long\def\P{\mathsf{FP}}%
\global\long\def\NP{\mathsf{NP}}%
\global\long\def\sharpWone{\mathsf{\#W[1]}}%
\global\long\def\Wone{\mathsf{W[1]}}%
\global\long\def\sharpSAT{\mathrm{\#SAT}}%

\global\long\def\leqFPT{\preceq_{\mathit{fpt}}^{T}}%
\global\long\def\leqP{\preceq_{\mathit{p}}^{T}}%
\global\long\def\abs#1{\left\Vert #1\right\Vert }%

\ytableausetup{smalltableaux}

\global\long\def\tSquare{\,\raisebox{4pt}{\scalebox{0.55}{\ydiagram{2,2}}}\,}%

\global\long\def\tSnake{\,\raisebox{5pt}{\scalebox{0.45}{\ydiagram{1+1,2,1}}}\,}%

\global\long\def\tLine{\,\raisebox{0pt}{\scalebox{0.55}{\ydiagram{4}}}\,}%

\global\long\def\tL{\,\raisebox{3pt}{\scalebox{0.55}{\ydiagram{3,1}}}\,}%

\global\long\def\tHDom{\,\raisebox{1pt}{\scalebox{0.55}{\ydiagram{2}}}\,}%

\global\long\def\tVDom{\,\raisebox{3pt}{\scalebox{0.55}{\ydiagram{1,1}}}\,}%

\global\long\def\makeRed#1{\textcolor{red}{#1}}%

\newcommand{\SSTabOne}{
\begin{ytableau} 
*(cyan)    & *(red)     & *(yellow)  & *(green) \\ 
*(cyan)    & *(red)     & *(yellow)  & *(green) \\ 
*(cyan)    & *(red)     & *(yellow) \\
*(red)     & *(red)     & *(yellow) \\
*(yellow)  & *(yellow)  & *(yellow) \\
*(green)   & *(green)   
\end{ytableau}
}

\newcommand{\SSTabTwo}{
\begin{ytableau} 
*(cyan)    & *(red)  	& *(green)	& *(green) \\ 
*(cyan)    & *(red) 	& *(green) & *(green) \\ 
*(yellow)  & *(yellow)  & *(green)  \\
*(yellow)  & *(yellow)  & *(green)  \\
*(green)   & *(green)   & *(green)  \\
*(green)   & *(green)   
\end{ytableau}
}

\newcommand{\SSTabThree}{
\begin{ytableau} 
*(cyan)    & *(red)     & *(red)     & *(green) \\ 
*(cyan)    & *(red)     & *(green)   & *(green) \\ 
*(cyan)    & *(red)     & *(green)  \\
*(yellow)  & *(yellow)  & *(green)  \\
*(yellow)  & *(green)   & *(green)  \\
*(green)   & *(green)   
\end{ytableau}
}

\newcommand{\SSLineOne}{
\begin{ytableau} 
*(cyan) \\
*(cyan) \\
*(cyan) \\
*(cyan)
\end{ytableau}
}

\newcommand{\SSLineTwo}{
\begin{ytableau} 
*(red) \\
*(red) \\
*(green) \\
*(green)
\end{ytableau}
}

\newcommand{\SSLineEmpty}{
\ydiagram{1,1,1,1}
}

\begin{abstract}
Given an integer $n\geq1$ and an irreducible character $\chi_{\lambda}$
of $S_{n}$ for some partition $\lambda$ of $n$, the immanant $\mathrm{imm}_{\lambda}:\mathbb{C}^{n\times n}\to\mathbb{C}$
maps matrices $A\in\mathbb{C}^{n\times n}$ to 
\[
\mathrm{imm}_{\lambda}(A)=\sum_{\pi\in S_{n}}\chi_{\lambda}(\pi)\prod_{i=1}^{n}A_{i,\pi(i)}.
\]
Important special cases include the \emph{determinant }and \emph{permanent},
which are the immanants associated with the \emph{sign} and \emph{trivial}
character, respectively.

It is known that immanants can be evaluated in polynomial time for
characters that are ``close'' to the sign character: Given a partition
$\lambda$ of $n$ with $s$ parts, let $b(\lambda):=n-s$ count the
boxes to the right of the first column in the Young diagram of $\lambda$.
For a family of partitions $\Lambda$, let $b(\Lambda):=\max_{\lambda\in\Lambda}b(\lambda)$
and write $\immProb(\Lambda)$ for the problem of evaluating $\imm_{\lambda}(A)$
on input $A$ and $\lambda\in\Lambda$.
\begin{itemize}
\item If $b(\Lambda)<\infty$, then $\immProb(\Lambda)$ is known to be
polynomial-time computable. This subsumes the case of the determinant.
\item If $b(\Lambda)=\infty$, then previously known hardness results suggest
that $\immProb(\Lambda)$ cannot be solved in polynomial time. However,
these results only address certain restricted classes of families
$\Lambda$.
\end{itemize}
In this paper, we show that the parameterized complexity assumption
$\FPT\neq\Wone$ rules out polynomial-time algorithms for $\immProb(\Lambda)$
for any computationally reasonable family of partitions $\Lambda$
with $b(\Lambda)=\infty$. We give an analogous result in algebraic
complexity under the assumption $\VFPT\neq\VW$. Furthermore, if $b(\lambda)$
even grows polynomially in $\Lambda$, we show that $\immProb(\Lambda)$
is hard for $\sharpP$ and $\VNP$. This concludes a series of partial
results on the complexity of immanants obtained over the last 35 years.
\end{abstract}
\newpage{}

\section{Introduction}

The determinant and permanent of an $n\times n$ matrix $X=(x_{i,j})$
can be defined by the sum-product formulas
\begin{align*}
\det(X) & =\sum_{\pi\in S_{n}}\mathrm{sgn}(\pi)\prod_{i=1}^{n}x_{i,\pi(i)},\\
\per(X) & =\sum_{\pi\in S_{n}}\prod_{i=1}^{n}x_{i,\pi(i)}.
\end{align*}
The similarity between these formulas is deceiving: While determinants
admit polynomial-size circuits and can be evaluated in polynomial
time, only exponential-size circuits and exponential-time algorithms
are known for permanents. Valiant~\cite{DBLP:journals/tcs/Valiant79}
underpinned this divide by proving that evaluating permanents is $\sharpP$-hard:
Any polynomial-time algorithm for this problem would entail a polynomial-time
algorithm for counting (and thus deciding the existence of) satisfying
assignments to Boolean formulas, thereby collapsing $\mathsf{P}$
and $\NP$\emph{.} With the $\VNP$-completeness of the permanent
family~\cite{DBLP:conf/stoc/Valiant79a}, an analogous statement
holds in algebraic complexity theory. 

Unconditional lower bounds for the complexity of permanents however
remain elusive, with only a quadratic lower bound on the determinantal
complexity of permanents known~\cite{8178470,DBLP:conf/stoc/CaiCL08}.
That is, expressing the permanent of an $n\times n$ matrix $X$ as
the determinant of an $m\times m$ matrix (whose entries are linear
forms in the entries of $X$) is known to require $m=\Omega(n^{2})$.
One of the core objectives in algebraic complexity theory lies in
proving that $m$ must grow super-polynomially~\cite{DBLP:conf/stoc/Valiant79a,DBLP:books/daglib/0090316,DBLP:books/daglib/0025071},
and this can be viewed as an algebraic version of the $\mathsf{P}\neq\NP$
problem.

\paragraph{The family of immanants.}

To understand the relationship between determinants and permanents
better, it may help to recognize them as part of a larger family:
The \emph{immanants} are matrix forms that are arranged on a spectrum
in which the determinant and permanent represent extreme cases. These
forms were studied by Schur~\cite{Schur1918,schur_thesis} in the
context of group character theory, and Littlewood and Richardson later
explicitly introduced them as immanants~\cite{10.2307/91293}.

Given any \emph{class function} $f:S_{n}\to\mathbb{C}$, i.e., a function
of permutations that depends only on the (multiset of) cycle lengths
of the input permutation, the immanant $\imm_{f}:\mathbb{C}^{n\times n}\to\mathbb{C}$
is defined by replacing the permutation sign $\mathrm{sgn}(\pi)$
in the determinant expansion with $f(\pi)$:

\[
\imm_{f}(X)=\sum_{\pi\in S_{n}}f(\pi)\prod_{i=1}^{n}x_{i,\pi(i)}.
\]

In the literature, immanants are typically defined by requiring $f$
to be an \emph{irreducible character} of $S_{n}$, i.e., an element
from a particular basis for the vector space of class functions.\footnote{The resulting immanants are sometimes also called \emph{character}
immanants, as opposed to other types of immanants, such as the Kazhdan-Lusztig
immanants~\cite{RHOADES2006793}.} General $f$-immanants can then be expressed as linear combinations
of such immanants. Two extremal examples of irreducible characters
are the \emph{trivial }character $\mathbf{1}:S_{n}\to\{1\}$ and the
\emph{sign} character $\mathrm{sgn}:S_{n}\to\{-1,1\}$, which induce
\begin{align*}
\det(X) & =\imm_{\mathrm{sgn}}(X),\\
\per(X) & =\imm_{\mathbf{1}}(X).
\end{align*}
In general, the irreducible characters of $S_{n}$ correspond naturally
to partitions of $n$, as outlined in Section~\ref{sec: rep-th}.
To see the existence of such a correspondence, note that the dimension
of the space of class functions on $S_{n}$ is the number of different
cycle length formats of $n$-permutations, that is, different partitions
of the integer $n$. For now, let us remark that the refinement-wise
minimal and maximal partitions $(1,\ldots,1)$ and $(n)$ naturally
correspond to the sign and trivial character, respectively. As another
example, we have
\begin{equation}
\chi_{(2,1,\ldots,1)}(\pi)=\mathrm{sgn}(\pi)\cdot(\#\{\text{fixed points of }\pi\}-1).\label{eq: chi_21111}
\end{equation}

Abbreviating $\imm_{\lambda}=\imm_{\chi_{\lambda}}$, we have $\det=\imm_{(1,\ldots,1)}$
and $\per=\imm_{(n)}$. Likewise, $\imm_{(2,1,\ldots,1)}$ sums over
row-column permutations of a matrix with weights as given in (\ref{eq: chi_21111}).
For a more applied example, it is known that the number of Hamiltonian
cycles in a directed $n$-vertex graph $G$, i.e., the immanant associated
with the indicator function for cyclic permutations (evaluated on
the adjacency matrix of $G$) can be written as a linear combination
of the \emph{hook} immanants $\imm_{(r,1^{n-r})}$ for $1\leq r\leq n$.

Beyond their theoretical origins in group character theory, immanants
have been applied in combinatorial chemistry~\cite{DBLP:journals/jcisd/Cash03}
and linear optics~\cite{PhysReview}, and they feature in (conjectured)
inequalities in matrix analysis~\cite{SHCHESNOVICH2016196}. In this
paper, we focus on complexity-theoretic aspects of immanants and their
role as an interpolating family between determinants and permanents.

\paragraph{The complexity of immanants.}

It is known that irreducible characters of $S_{n}$ can be evaluated
in polynomial time~\cite{DBLP:journals/siamcomp/Burgisser00a,DBLP:books/daglib/0025071}.
Using this, any character immanant of an $n\times n$ matrix can be
evaluated in $n!\cdot n^{O(1)}$ time by brute-force, or in $2^{n+o(n)}$
time by a variant of the Bellman--Held--Karp dynamic programming
approach for Hamiltonian cycles. For some immanants however, among
them the determinant, this exponential running time is far from optimal:
Hartmann~\cite{doi:10.1080/03081088508817680} gave an algorithm
for evaluating $\imm_{\lambda}$ in $O(n^{6b(\lambda)+4})$ time,
where $b(\lambda):=n-s$ for a partition $\lambda=(\lambda_{1},\ldots,\lambda_{s})$
with $s$ parts. In visual terms, the quantity $b(\lambda)$ counts
the boxes to the right of the first column in the Young diagram of
$\lambda$, which is a left-aligned shape whose $i$-th row contains
$\lambda_{i}$ boxes, when $\lambda$ is ordered non-increasingly:
\begin{table}[H]
\centering{}%
\begin{tabular}{c}
\ydiagram{4,4,3,3,3,2}\tabularnewline
\tabularnewline
$\lambda=(4,4,3,3,3,2)\text{ with }b(\lambda)=13$\tabularnewline
\end{tabular}
\end{table}

In particular, partitions $\lambda$ with $b(\lambda)=O(1)$ induce
polynomial-time solvable immanants. Barvinok~\cite{barvinok} and
Bürgisser~\cite{DBLP:journals/siamcomp/Burgisser00} later gave algorithms
with improved running times $O(n^{2}d_{\lambda}^{4})$ and $O(n^{2}s_{\lambda}d_{\lambda})$,
where $s_{\lambda}$ and $d_{\lambda}$ denote the numbers of \emph{standard}
and \emph{semi-standard} tableaux of shape $\lambda$.\footnote{Given a partition $\lambda$ of $n$, a standard tableau of shape
$\lambda$ is an assignment of the numbers $1,\ldots,n$ to the boxes
in the Young diagram of $\lambda$ such that all rows and columns
are strictly increasing. In a semi-standard tableau, rows are only
required to be non-decreasing.} These algorithms give better running times in the exponential-time
regime, but they do not identify new polynomial-time solvable immanants.
One is therefore naturally led to wonder whether $b(\lambda)$ is
indeed the determining parameter for the complexity of immanants.
To investigate this formally, we consider \emph{families} of partitions
$\Lambda$ and define $\immProb(\Lambda)$ as the problem of evaluating
$\imm_{\lambda}(A)$ on input a matrix $A$ and a partition $\lambda\in\Lambda$.
As discussed above, the problem $\immProb(\Lambda)$ is polynomial-time
solvable if the quantity 
\[
b(\Lambda):=\max_{\lambda\in\Lambda}b(\lambda)
\]
is finite. On the other hand, for various families $\Lambda$ with
unbounded $b(\Lambda)$, the problem $\immProb(\Lambda)$ is indeed
known to be hard for the counting complexity class $\sharpP$ and
its algebraic analog $\VNP$:
\begin{itemize}
\item Bürgisser~\cite{DBLP:journals/siamcomp/Burgisser00a} showed $\VNP$-completeness
and $\sharpP$-hardness of $\immProb(\Lambda)$ for any family $\Lambda$
of hook partitions $(t(n),1^{n-t(n)})$, provided that $t=\Omega(n^{\alpha})$
with $\alpha>0$ can be computed in polynomial time. A similar result
appears in Hartmann's work~\cite{doi:10.1080/03081088508817680}.
\item In the same paper, Bürgisser showed similar hardness results for families
of rectangular partitions of polynomial width. (The width is the largest
entry in the partition.)
\end{itemize}
In his 2000 monograph~\cite{DBLP:books/daglib/0025071}, Bürgisser
conjectures that $\immProb(\Lambda)$ is hard for \emph{any} reasonable
family $\Lambda$ of polynomial width. He also asks about the complexity
status of partitions of width $2$, and overall deems the complexity
of immanants to be ``still full of mysteries''. Some of these mysteries
have since been resolved:
\begin{itemize}
\item In 2003, Brylinski and Brylinski~\cite{DBLP:journals/corr/cs-CC-0301024}
showed $\VNP$-completeness for any family of partitions $\Lambda$
with a gap of width $\Omega(n^{\alpha})$ for $\alpha>0$. Here, a
gap is the difference between two consecutive rows.
\item In 2013, Mertens and Moore~\cite{DBLP:journals/toc/MertensM13} proved
$\sharpP$-hardness for the family $\Lambda$ of \emph{all} partitions
of width $2$, that is, the partitions containing only entries $1$
and $2$. They also proved $\mathsf{\oplus P}$-hardness for the more
restricted family of partitions containing only the entry $2$.
\item In the same year, de Rugy-Altherre~\cite{DBLP:conf/cie/Rugy-Altherre13}
gave a dichotomy for partition families $\Lambda$ of constant width
and polynomial growth of $b(\lambda)$, confirming for such families
that boundedness of $b(\Lambda)$ indeed determines the complexity
of $\immProb(\Lambda)$.
\end{itemize}
However, an exhaustive complexity classification of $\immProb(\Lambda)$
for general partition families $\Lambda$ still remained open, even
35 years after Hartmann's initial paper~\cite{doi:10.1080/03081088508817680}
and despite several appearances as an open problem~\cite{DBLP:journals/toc/MertensM13,DBLP:conf/cie/Rugy-Altherre13},
also in a monograph~\cite{DBLP:books/daglib/0025071}. In fact, even
very special cases like $\immProb(\Lambda)$ for the staircase partitions
$(k,k-1,\ldots,1)$ remained unresolved~\cite{DBLP:conf/cie/Rugy-Altherre13}.

\subsection{Our results\label{sec: intro-results}}

We classify the complexity of the problems $\immProb(\Lambda)$ for
partition families $\Lambda$ satisfying natural computability and
density conditions that are satisfied by all families studied in the
literature. Under the assumption $\FPT\neq\sharpWone$ from parameterized
complexity~\cite{DBLP:journals/siamcomp/FlumG04}, we confirm that
$\immProb(\Lambda)$ is polynomial-time solvable iff $b(\Lambda)$
is unbounded. An algebraic analogue holds under the assumption $\VFPT\neq\VW$
introduced by Bläser and Engels~\cite{DBLP:conf/iwpec/BlaserE19}.
(Please consider Section~\ref{subsec: prelim-complexity} for a brief
introduction to the relevant complexity classes.) Our classification
holds even if $b(\lambda)$ only grows sub-polynomially in $\Lambda$,
which allows us to address families such as 
\begin{equation}
\Lambda_{\mathrm{log}}=\{(\lceil\log n\rceil,1^{n})\mid n\in\mathbb{N}\}.\label{eq: Lambda_log}
\end{equation}
Note that $\immProb(\Lambda_{\mathrm{log}})$ can be solved in $n^{O(\log n)}$
time by the $n^{O(b(\lambda))}$ time algorithms discussed before,
which likely prevents hardness for $\sharpP$ or $\VNP$. At the same
time, a polynomial-time algorithm seems unlikely. Thus, partition
families like $\immProb(\Lambda_{\mathrm{log}})$ fall into the ``blind
spot'' of classical dichotomies, an issue that is also alluded to
in~\cite{DBLP:conf/cie/Rugy-Altherre13}.

Our sanity requirements on $\Lambda$ are encapsulated as follows:
We say that $\Lambda$ \emph{supports growth }$g:\mathbb{N}\to\mathbb{N}$
if every $n\in\mathbb{N}$ admits a partition $\lambda^{(n)}\in\Lambda$
with $b(\lambda^{(n)})\geq g(n)$ and total size $\Theta(n)$. This
ensures that $\Lambda$ is dense enough and that $\Lambda$ supplies
sufficiently many boxes both in the first column \emph{and} to the
right of it. We may also require that $\lambda^{(n)}$ can be computed
in polynomial time on input $n\in\mathbb{N}$ and then say that $\Lambda$
\emph{computationally }supports growth $g$. This condition is not
required for the algebraic completeness results.
\begin{example*}
The family of staircase partitions $(n,n-1,\ldots,1)$ supports growth
$\Omega(n)$. The partition families $(\lceil\log n\rceil,1^{n})$
and $(n,1^{2^{n}})$ for $n\in\mathbb{N}$ support growth $\Omega(\log n)$,
even though the second family is exponentially sparse. On the other
hand, partition families whose sizes grow \emph{doubly} exponentially
do not support any growth by our definition. It might still be possible
to address such families via ``infinitely often'' versions of $\sharpP$
or $\VNP$, but we currently see no added value in doing so.
\end{example*}
In the polynomial growth regime for $b(\lambda)$, we obtain classical
$\sharpP$-hardness and $\VNP$-completeness results. As a bonus,
we also obtain the expected quantitative lower bounds under the exponential-time
hypothesis $\sharpETH$, which postulates that counting satisfying
assignments to $n$-variable $3$-CNFs takes $\exp(\Omega(n))$ time.
\begin{thm}
\label{thm: main-poly}For any family of partitions $\Lambda$:
\begin{itemize}
\item If $b(\Lambda)<\infty$, then $\immProb(\Lambda)\in\P$ and $\immProb(\Lambda)\in\VP$.
\item Otherwise, if $\Lambda$ supports growth $\Omega(n^{\alpha})$ for
some $\alpha>0$, then $\immProb(\Lambda)$ is $\VNP$-complete. If
$\Lambda$ computationally supports growth $\Omega(n^{\alpha})$,
then $\immProb(\Lambda)$ is $\sharpP$-hard and admits no $\exp(o(n^{\alpha}))$
time algorithm unless $\sharpETH$ fails. 
\end{itemize}
\end{thm}

Theorem~\ref{thm: main-poly} subsumes all known $\VNP$-hardness
and $\sharpP$-hardness results for immanant families, confirms the
conjecture from \cite{DBLP:journals/toc/MertensM13}, and settles
the case of staircases. Using parameterized complexity theory, we
also address the sub-polynomial growth regime for $b(\lambda)$. To
this end, we consider \emph{parameterized problems}, whose instances
$(x,k)$ come with a \emph{parameter} $k$. The corresponding objects
in the algebraic setting are \emph{parameterized polynomial families}
$(p_{n,k})$, where the second index $k$ is a parameter. A parameterized
problem (or polynomial family) is said to be \emph{fixed-parameter
tractable} if it can be solved in $f(k)\cdot n{}^{O(1)}$ time with
$n=|x|$ (or admits circuits of that size) for some computable function
$f$. The problem (or polynomial family) is then said to be contained
in $\FPT$ (or $\VFPT$). The classes $\sharpWone\supseteq\FPT$ (and
$\VW\supseteq\VFPT$) contain problems (and polynomial families) that
are believed not to be fixed-parameter tractable.
\begin{thm}
\label{thm: main-param}For any family of partitions $\Lambda$:
\begin{itemize}
\item If $b(\Lambda)<\infty$, then $\immProb(\Lambda)\in\P$ and $\immProb(\Lambda)\in\VP$.
\item Otherwise, if $\Lambda$ supports growth $g\in\omega(1)$, then $\immProb(\Lambda)\notin\VP$
unless $\VFPT=\VW$. If $\Lambda$ computationally supports growth
$g$, then $\immProb(\Lambda)\notin\P$ unless $\FPT=\sharpWone$.
\end{itemize}
\end{thm}

Note that we do not prove $\immProb(\Lambda)$ to be \emph{hard} for
$\sharpWone$ or $\VW$ as a parameterized problem with parameter
$b(\lambda)$, even though this might seem natural. Indeed, problems
like $\immProb(\Lambda_{\mathrm{log}})$ are trivially fixed-parameter
tractable in the parameter $b(\lambda)$. We only show that \emph{polynomial-time}
algorithms for $\immProb(\Lambda)$ would render $\sharpWone$-hard
problems fixed-parameter tractable.

\subsection{Proof outline\label{sec: intro:proof}}

We establish Theorems~\ref{thm: main-poly} and \ref{thm: main-param}
by reduction from $\match$, the problem of counting $k$-matchings
in bipartite graphs $H$. When parameterized by $k$, this problem
is $\sharpWone$-complete~\cite{DBLP:conf/iwpec/BlaserC12,DBLP:conf/icalp/Curticapean13,DBLP:conf/focs/CurticapeanM14,DBLP:conf/stoc/CurticapeanDM17},
with an analogous statement in the algebraic setting~\cite{DBLP:conf/iwpec/BlaserE19}.
When $k$ grows polynomially in $|V(H)|$, counting $k$-matchings
is complete for $\sharpP$ and $\VNP$ by a trivial reduction from
the permanent~\cite{DBLP:journals/tcs/Valiant79}.

To reduce counting matchings to immanants, we proceed in three stages:
First, we identify two types of ``exploitable resources'' in partitions,
then we show how to exploit them for a reduction, and finally we wrap
the proof up in complexity-theoretic terms.

\subsubsection*{Extracting resources (Section~\ref{sec:Staircases-versus-skew})}

Our construction relies on two types of resources that can supplied
by a given partition: A large \emph{staircase }or a large number of
\emph{non-vanishing tetrominos}.

\begin{figure}[H]
\centering{}\includegraphics[width=10cm]{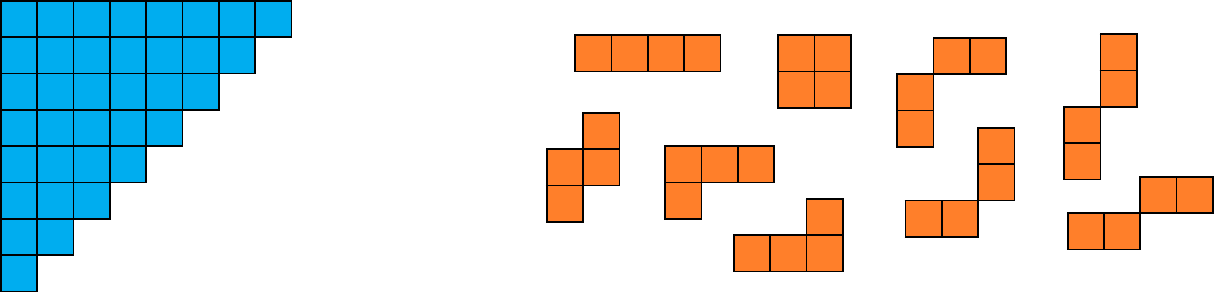}\label{fig: tetrofig}
\end{figure}
To define these notions, consider successively ``peeling'' dominos
$\tHDom$ and $\tVDom$ from $\lambda$, that is, removing them from
the south-eastern border of $\lambda$ while ensuring that the shape
obtained after each step has non-increasing row lengths. After peeling
the maximum number of dominos this way, we reach some (possibly empty)
staircase $\mu$, which is easily seen to be unique. The \emph{domino
number} $d(\lambda)$ is this maximum number of removable dominos,
and $\mu$ is the \emph{staircase of} $\lambda$; we write $w(\lambda)$
for its width.

Now consider peeling \emph{two} dominos from a partition $\lambda$.
Some of the shapes that can arise this way are shown above in orange;
we call them ``non-vanishing tetrominos'' for reasons that will
become evident in the proof. For the four corner-connected domino
pairs, we adopt the convention that their two dominos must be peeled
successively from disjoint rows and columns of $\lambda$. The\emph{
non-vanishing tetromino number} $s(\lambda)$ then is the maximum
number of non-vanishing tetrominos that can be peeled from $\lambda$.
Note that this number is $0$ for the determinant-inducing partition
$(1,\ldots,1)$ and $\lfloor n/4\rfloor$ for the permanent-inducing
partition $(n)$.

In Section~\ref{sec:Staircases-versus-skew}, we establish a ``win-win
situation'' for these resources: For any partition $\lambda$, at
least one of $w(\lambda)\in\Omega(\sqrt{b(\lambda)})$ or $s(\lambda)\in\Omega(b(\lambda))$
must hold.

\subsubsection*{Exploiting resources (Sections~\ref{sec: staircase} and \ref{sec: tetrominos})}

Next, we outline how to exploit staircases and non-vanishing tetrominos
in a partition $\lambda$ for reductions from counting $k$-matchings
to evaluating $\imm_{\lambda}$. Throughout this paper, the immanant
of a directed graph $G$ refers to the immanant of its adjacency matrix
$A$. This way, we can view immanants as character-weighted sums over
the cycle covers of digraphs.

Given an $n$-vertex graph $H$ and $k\in\mathbb{N}$, we construct
a digraph $G$ such that $\imm_{\lambda}(G)$ counts the $k$-matchings
in $H$ up to a constant factor $c_{\lambda,k}$ that can be computed
in polynomial time. In a second step, we show that the constant factor
$c_{\lambda,k}$ is non-zero if $\lambda$ supplies enough resources.
Both the construction of $G$ and the handling of $c_{\lambda,k}$
differ for staircases and tetrominos, as outlined below.

\paragraph{Non-vanishing tetrominos (Section~\ref{sec: tetrominos})}

Each non-vanishing tetromino peeled from $\lambda$ enables a particular
\emph{edge gadget}: To count the $k$-matchings in a graph $H$, we
replace each edge $uv\in E(H)$ by the gadget shown below; the weight
$w$ of $uv$ appears on two edges of the gadget. Together with additional
constructions detailed in Section~\ref{sec: tetrominos}, this results
in a directed graph $G$. 

\begin{figure}[h]
\centering \def\svgwidth{7cm} 
\begingroup%
  \makeatletter%
  \providecommand\color[2][]{%
    \errmessage{(Inkscape) Color is used for the text in Inkscape, but the package 'color.sty' is not loaded}%
    \renewcommand\color[2][]{}%
  }%
  \providecommand\transparent[1]{%
    \errmessage{(Inkscape) Transparency is used (non-zero) for the text in Inkscape, but the package 'transparent.sty' is not loaded}%
    \renewcommand\transparent[1]{}%
  }%
  \providecommand\rotatebox[2]{#2}%
  \newcommand*\fsize{\dimexpr\f@size pt\relax}%
  \newcommand*\lineheight[1]{\fontsize{\fsize}{#1\fsize}\selectfont}%
  \ifx\svgwidth\undefined%
    \setlength{\unitlength}{100.69920217bp}%
    \ifx\svgscale\undefined%
      \relax%
    \else%
      \setlength{\unitlength}{\unitlength * \real{\svgscale}}%
    \fi%
  \else%
    \setlength{\unitlength}{\svgwidth}%
  \fi%
  \global\let\svgwidth\undefined%
  \global\let\svgscale\undefined%
  \makeatother%
  \begin{picture}(1,0.35686507)%
    \lineheight{1}%
    \setlength\tabcolsep{0pt}%
    \put(0,0){\includegraphics[width=\unitlength,page=1]{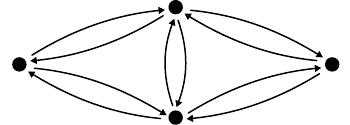}}%
    \put(0.54254466,0.16803972){\color[rgb]{0,0,0}\makebox(0,0)[lt]{\lineheight{1.25}\smash{\begin{tabular}[t]{l}$-1$\end{tabular}}}}%
    \put(0.0310358,0.23775458){\color[rgb]{0,0,0}\makebox(0,0)[lt]{\lineheight{1.25}\smash{\begin{tabular}[t]{l}$u$\end{tabular}}}}%
    \put(0.94132248,0.23775458){\color[rgb]{0,0,0}\makebox(0,0)[lt]{\lineheight{1.25}\smash{\begin{tabular}[t]{l}$v$\end{tabular}}}}%
    \put(0.20687899,0.30294607){\color[rgb]{0,0,0}\makebox(0,0)[lt]{\lineheight{1.25}\smash{\begin{tabular}[t]{l}$w$\end{tabular}}}}%
    \put(0.29574506,0.13272596){\color[rgb]{0,0,0}\makebox(0,0)[lt]{\lineheight{1.25}\smash{\begin{tabular}[t]{l}$w$\end{tabular}}}}%
  \end{picture}%
\endgroup%
\label{fig: eq-gadget}
\end{figure}
The edge gadget effectively constrains the set of cycle covers counted
by the immanant, as undesired cycle covers cancel out in pairs of
opposite signs. In the remaining cycle covers of $G$, each gadget
is either in the \emph{passive state} (shown below in cyan) or in
one of four \emph{active states }(two are shown below in green, two
more are symmetric versions thereof). 

\begin{figure}[h]
\centering \def\svgwidth{12cm} 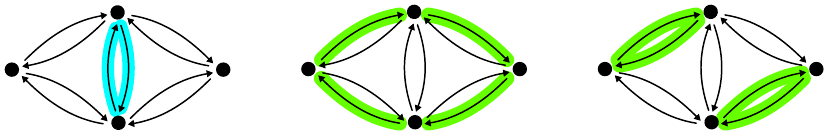
\end{figure}

This allows us to simulate matchings $M$ in $H$ via cycle covers
in $G$: We interpret active gadgets as matching edges $e\in M$ and
passive gadgets as edges $e\in E(H)\setminus M$. Intuitively speaking,
each active gadget ``uses up'' one non-vanishing tetromino of $\lambda$,
while passive gadgets only require a domino. Overall, if we can peel
$O(k)$ non-vanishing tetrominos and some number of dominos from $\lambda$,
then $\imm_{\lambda}$ can be used to count $k$-matchings in $H$.

\paragraph{Large staircase (Section~\ref{sec: staircase})}

If the staircase $\mu$ of $\lambda$ contains $\Omega(k)$ boxes,
then properties of the staircase character $\chi_{\mu}$ enable an
ad-hoc reduction from counting $k$-matchings in bipartite graphs
to the $\lambda$-immanant. More specifically, we observe and use
that cycle covers with even cycles vanish in staircase characters
$\chi_{\mu}$. After discarding irrelevant dominos from $\lambda$,
we can then use this fact together with a particular graph construction
to compute a sum over cycle covers with one \emph{particular fixed
}cycle length format by reduction to $\imm_{\lambda}$. This way of
exploiting staircase characters to avoid even-length cycles may also
be useful for other algorithmic applications.

\subsubsection*{Wrap-up (Section~\ref{sec:Completing-the-proof})}

For a streamlined presentation, the two reductions above are encapsulated
as mere mathematical formulas relating the number of $k$-matchings
in a graph $H$ with the immanant of a digraph $G$ constructed from
$H$. In Section~\ref{sec:Completing-the-proof}, we add the necessary
``wrapper code'' to obtain the (parameterized and polynomial-time,
algebraic and computational) reductions required to prove Theorems~\ref{thm: main-poly}
and \ref{thm: main-param}.

\subsection{Proof highlights}

Most of our arguments rely on making non-vanishing tetrominos and
staircases come together in\emph{ }just the right way. Additionally,
the following can be pointed out:
\begin{itemize}
\item In the tetromino-based reduction, the particular form of active and
passive states in edge gadgets ensures that we only need to understand
character values $\chi_{\lambda}(\rho)$ on cycle formats $\rho$
with cycle lengths $1$, $2$, and $4$. This allows us to sidestep
more involved representation-theoretic arguments that occur in related
works. It should be noted that \emph{equality} and \emph{exclusive-or}
gadgets that enforce consistency properties of cycle covers are common
in algebraic and counting complexity, dating back to Valiant~\cite{DBLP:conf/stoc/Valiant79a,DBLP:journals/tcs/Valiant79}.
The particular idea of repurposing an equality gadget into an \emph{edge}
gadget was also already used in the author's very first paper with
Bläser~\cite{DBLP:conf/mfcs/BlaserC11}.
\item Parameterized complexity assumptions allow us to handle cases that
cannot be addressed in classical frameworks, such as the family $\immProb(\Lambda_{\mathrm{log}})$
discussed before. By basing hardness on the assumptions $\FPT\neq\sharpWone$
and $\VFPT\neq\VW$, we can still argue about such families.
\item Two inconspicuous but crucial proof steps (Lemmas~\ref{lem: domino-tiling-sign}
and Fact~\ref{fact: nonvan-tetro}) rely on cute arguments involving
``dominos on chessboards'' that one would rather expect in the context
of recreational mathematics. This can again be credited to the edge
gadget, as such arguments would likely fail for cycle lengths other
than $2$ and $4$.
\end{itemize}

\section{Preliminaries}

We start with basic definitions for graphs and partitions in Section~\ref{subsec:Basic-notions}.
In Sections~\ref{subsec: prelim-complexity} and \ref{subsec:Counting-matchings},
we present the complexity-theoretic preliminaries used in this paper.
A bare-bones introduction to the relevant character theory of symmetric
groups can be found in Section~\ref{sec: rep-th}.

\subsection{Basic notions\label{subsec:Basic-notions}}

\paragraph{Graphs.}

We consider undirected graphs (when counting matchings) and directed
graphs (when evaluating immanants). Both graph types may feature (indeterminate)
edge-weights, and directed graphs may feature self-loops. For a graph
$G$ with adjacency matrix $A$, we write $\imm(G)$ instead of $\imm(A)$
and view the immanant as a sum over cycle covers: A \textbf{partial
cycle cover}\emph{ }in a directed graph $G$ is an edge-set $C\subseteq E(G)$
such that each vertex has one incoming and one outgoing edge in $C$.
It is a \textbf{cycle cover} if all of $V(G)$ is covered. Directed
$2$-cycles will also be called \textbf{digons}\emph{.}

\paragraph{Partitions.}

A \textbf{partition} of a positive integer $n\in\mathbb{N}$ is a
multi-set $\lambda$ of positive integers such that $\sum_{i\in\lambda}i=n$.
Its elements are called \textbf{parts}, and we write $\lambda\vdash n$
to indicate that $\lambda$ is a partition of $n$. Several notations
will be used for partitions:
\begin{center}
\begin{tabular}{ccccc}
sorted tuple &  & compact notation &  & Young diagram\tabularnewline
\vspace{-0.3cm}
 &  &  &  & \tabularnewline
$(4,4,3,3,3,2)$ &  & $(4^{2},3^{3},2^{1})$ &  & \ydiagram{4,4,3,3,3,2}\tabularnewline
\vspace{-0.2cm}
 &  &  &  & \tabularnewline
\end{tabular}
\par\end{center}

As depicted above, the \textbf{Young diagram} of a partition $\lambda=(\lambda_{1},\ldots,\lambda_{s})$
is a left-aligned shape consisting of $\lambda_{i}$ boxes in row
$i$. We define the \textbf{gap} $\delta_{i}$ of row $i$ as $\delta_{i}:=\lambda_{i}-\lambda_{i+1}$,
where we consider $\lambda_{s+1}:=0$.

Given partitions $\lambda\vdash n$ and $\lambda'\vdash n'$, we sometimes
abuse notation and write $(\lambda,\lambda')$ for the partition $\lambda\cup\lambda'$
of $n+n'$. Sometimes we also specify an ordering for the elements
in a partition: An \textbf{ordered partition} of $n$ (also called
\textbf{composition}) is a tuple of positive integers that sum to
$n$. We will state explicitly when partitions are considered to be
ordered.

Let us stress that any permutation $\pi$ of $n$ elements (and any
cycle cover $C$ of an $n$-vertex graph) naturally induces a partition
$\lambda\vdash n$ through its cycle lengths, the \textbf{cycle format}
$\rho(\pi)$ of $\pi$. For example, the identity permutation has
cycle format $(1,\ldots,1)$ and cyclic permutations have cycle format
$(n)$. We also say that a cycle cover is a $\rho$-cycle cover if
its format is $\rho$.

\subsection{Complexity-theoretic notions\label{subsec: prelim-complexity}}

We recall basic notions from complexity theory. For a more comprehensive
overview, please consider \cite{DBLP:journals/tcs/Valiant79,DBLP:journals/siamcomp/FlumG04,DBLP:phd/dnb/Curticapean15}
for (parameterized) counting complexity and \cite{DBLP:conf/stoc/Valiant79a,DBLP:conf/iwpec/BlaserE19,DBLP:books/daglib/0090316,DBLP:books/daglib/0025071}
for (parameterized) algebraic complexity.

\paragraph{Counting complexity.}

We view functions $f:\{0,1\}^{*}\to\mathbb{Q}$ as \textbf{counting
problems}. For example, the problem $\sharpSAT$ maps (binary encodings
of) Boolean formulas $\varphi$ to the number of satisfying assignments
in $\varphi$. Properly encoded, the permanent of rational-valued
matrices is a counting problem. A counting problem is contained in
$\P$ if it can be solved in polynomial time.

A \textbf{polynomial-time Turing reduction }from a counting problem
$\mathrm{\#A}$ to another counting problem $\mathrm{\#B}$ is a polynomial-time
algorithm that solves $\mathrm{\#A}$ with an oracle for $\mathrm{\#B}$.
We say that $\mathrm{\#B}$ is $\sharpP$\textbf{-hard} (under Turing
reductions) if $\sharpSAT$ admits a polynomial-time Turing reduction
to $\mathrm{\#B}$. While more stringent reduction notions exist,
they are not relevant for the purposes of this paper. Assuming $\P\neq\sharpP$,
no $\sharpP$-hard problem can be solved in polynomial time.

A \textbf{parameterized counting problem} features inputs $(x,k)$
for $x\in\{0,1\}^{*}$ and $k\in\mathbb{N}$. It is \textbf{fixed-parameter
tractable} if it can be solved in time $f(k)\cdot|x|^{O(1)}$ for
some computable function $f$, and we write $\FPT$ for the class
of such problems. A \textbf{parameterized Turing reduction} from a
parameterized counting problem $\mathrm{\#A}$ to another parameterized
problem $\mathrm{\#B}$ is an algorithm that solves any instance $(x,k)$
for $\mathrm{\#A}$ in time $f(k)\cdot|x|^{O(1)}$ with an oracle
for $\mathrm{\#B}$ that is only called on instances $(y,k')$ with
$k'\leq g(k)$. Here, both $f$ and $g$ are computable functions.
We say that $\mathrm{\#B}$ is \emph{$\sharpWone$}\textbf{-hard}
if the problem of counting $k$-cliques in a graph admits a parameterized
Turing reduction to $\mathrm{\#B}$. Assuming $\FPT\neq\sharpWone$,
no $\sharpWone$-hard problem is fixed-parameter tractable.

The (counting version of the) \textbf{exponential-time hypothesis}
$\sharpETH$ postulates that no $\exp(o(n))$ time algorithm solves
$\sharpSAT$ on $n$-variable formulas. If a problem $\mathrm{\#A}$
cannot be solved in $\exp(o(n))$ time, and it admits a polynomial-time
Turing reduction to a problem $\mathrm{\#B}$ such that every invoked
oracle query has size $O(n^{c})$ for $c\geq0$, then $\mathrm{\#B}$
cannot be solved in $\exp(o(n^{1/c}))$ time. Likewise, if a parameterized
problem $\mathrm{\#A}$ cannot be solved in $f(k)\cdot n^{o(k)}$
time and it admits a parameterized Turing reduction to a problem $\mathrm{\#B}$
such that every invoked query $(y,k')$ satisfies $k'\leq O(k)$,
then $\mathrm{\#B}$ also cannot be solved in $f(k)\cdot n^{o(k)}$
time. The hypothesis $\sharpETH$ rules out an $f(k)\cdot n^{o(k)}$
time algorithm for counting $k$-cliques and thus implies $\FPT\neq\sharpWone$.

\paragraph{Algebraic complexity.}

In the algebraic setting, \textbf{p-families} play the role of counting
problems: A sequence of multivariate polynomials $f=(f_{1},f_{2},\ldots)$
over some field is a p-family\emph{ }if, for all $n\in\mathbb{N}$,
the degree and the number of variables of $f_{n}$ are bounded by
$n^{O(1)}$. For example, the sequences of determinants and permanents
of $n\times n$ matrices with indeterminates are p-families. A p-family
is contained in $\VP$ if it admits a polynomial-size arithmetic circuit.

A p-family $f$ admits a \textbf{c-reduction} to another p-family
$g$ if there is an arithmetic circuit of polynomial size that computes
each $f_{n}$ with oracle gates for $g_{1},\ldots g_{n^{O(1)}}$.
A p-family $g$ is \emph{$\VNP$}\textbf{-hard}\emph{ }if the permanent
family admits a c-reduction to $g$. Assuming $\VP\neq\VNP$, no $\VNP$-hard
family is contained in $\VP$.

A \textbf{parameterized p-family} is a family $f=(f_{n,k})_{n,k\in\mathbb{N}}$
with two indices such that $f_{n,k}$ uses $n^{O(1)}$ variables and
has degree $(n+k)^{O(1)}$. A family $f$ is contained in $\VFPT$
if it admits an arithmetic circuit of size $h(k)\cdot n^{O(1)}$ for
some function $h$. An example for a parameterized p-family is given
by the \textbf{partial permanents} $\per_{n,k}=\sum_{\pi\in S(n,k)}\prod_{i=1}^{n}x_{i,\pi(i)}$
for $n,k\in\mathbb{N}$, where $S(n,k)$ is the set of all permutations
with $n-k$ fixed points. A parameterized family $f=(f_{n,k})_{n,k\in\mathbb{N}}$
admits a \textbf{parameterized c-reduction} to another family $g=(g_{n,k})_{n,k\in\mathbb{N}}$
if there is an arithmetic circuit of size $h(k)\cdot n^{O(1)}$ that
computes $f_{n,k}$ with oracle gates for polynomials $g_{n',k'}$
satisfying $n'\leq h(k)\cdot n^{O(1)}$ and $k'\leq h(k)$ for some
function $h$. We say that $g$ is \emph{$\VW$}\textbf{-hard} if
the partial permanents admit a parameterized c-reduction to $g$.
Assuming $\VFPT\neq\VW$, no $\VW$-hard problem is contained in $\VFPT$.

\subsection{Counting matchings\label{subsec:Counting-matchings}}

A \textbf{matching} in an undirected graph $H$ is a set $M\subseteq E(H)$
of pairwise disjoint edges. We write $\mathcal{M}_{k}(H)$ for the
set of matchings with $k$ edges in $H$. A matching $M$ is a \textbf{perfect
matching }if every vertex $v\in V(H)$ is contained in some edge of
$M$. Given an $n$-vertex graph $H$ with edge-weights $w:E(H)\to\mathbb{\mathbb{Q}}$,
we define 
\begin{align*}
\match(H,k) & =\sum_{M\in\mathcal{M}_{k}(H)\,}\prod_{e\in M}w(e),\\
\PerfMatch(H) & =\match(H,|V(H)|/2).
\end{align*}
Note that $\PerfMatch(H)$ is only defined for graphs $H$ with an
even number of vertices. We prove hardness of immanants by reduction
from (restrictions of) the problem of evaluating $\match$ on bipartite
graphs.
\begin{defn}
The counting problem $\PerfMatch$ asks to compute $\PerfMatch(H)$
for \emph{bipartite} graphs $H$ with edge-weights $w:E(H)\to\mathbb{Q}$. 

For any fixed polynomial-time computable function $g:\mathbb{N}\to\mathbb{N}$,
we define the counting problem $\match^{(g)}$: Given a pair $(H,k)$
consisting of an undirected bipartite graph $H$ with edge-weights
$w:E(H)\to\mathbb{Q}$ and a number $k\leq g(|V(H)|)$, compute $\match(H,k)$.
\end{defn}

On the complete bipartite graphs $K_{n,n}$ with indeterminate edge-weights,
$\PerfMatch$ induces the p-family of permanent polynomials, and $\match^{(g)}$
likewise induces a restriction of the partial permanent family with
$\match_{n,k}^{(g)}=0$ for $k>g(n)$. Abusing notation, we call these
p-families $\PerfMatch$ and $\match^{(g)}$ as well. Note that no
graphs are given as inputs to these families; the numbers of $k$-matchings
for given bipartite graphs $H$ can be obtained by evaluating the
polynomials at points whose non-zero coordinates encode the edges
of $H$.

The hardness results for $\PerfMatch$ and $\match^{(g)}$ required
in the remainder of the paper are either known in the literature or
can be derived easily. We collect the relevant results below.
\begin{thm}
\label{thm: match-bit-hardness}The following holds:
\begin{enumerate}
\item The problem $\PerfMatch$ is $\VNP$-complete and $\sharpP$-hard
and admits no $2^{o(n)}$ time algorithm under $\sharpETH$, even
on bipartite graphs of maximum degree $3$.
\item For any unbounded and polynomial-time computable function $g$, the
problem $\match^{(g)}$ is $\sharpWone$-hard and $\VW$-complete.
\end{enumerate}
\end{thm}

\begin{proof}
See \cite{DBLP:journals/tcs/Valiant79,DBLP:conf/stoc/Valiant79a}
for the $\sharpP$-hardness and $\VNP$-completeness of $\PerfMatch$
and \cite{DBLP:conf/icalp/Curticapean15} for the lower bound under
$\sharpETH$. The $\sharpWone$-hardness of the cardinality-unrestricted
problem $\match$ is shown in \cite{DBLP:conf/focs/CurticapeanM14},
and and $\VW$-completeness is shown in \cite[Lemma 8.7]{DBLP:conf/iwpec/BlaserE19}.

For any polynomial-time computable function $g$, we give a parameterized
reduction from $\match$ to $\match^{(g)}$: Any instance $(H,k)$
with an $n$-vertex graph and $k\leq g(n)$ can be solved directly
with a call to $\match^{(g)}$. On the other hand, if $k>g(n)$, then
we have $n<g^{-1}(k)$, so we can count $k$-matchings in $H$ by
brute-force in $g'(k)$ time for some computable function $g'$. This
satisfies the requirements of a parameterized reduction. In the algebraic
setting, we need not branch on $g(n)$ within a circuit, but instead
hard-code into the circuit family realizing $\match$ whether to (i)
perform a brute-force sum over matchings, or (ii) use the oracle gate
for $\match^{(g)}$.
\end{proof}

\section{Characters of the symmetric group\label{sec: rep-th}}

We give a minimal introduction to representations and characters of
the symmetric group; this material is covered thoroughly in classical
textbooks~\cite{fulton1991representation,macdonald1998symmetric,sagan2013symmetric,lorenz2018tour}.
For our purposes, a \textbf{representation} of $S_{n}$ is a homomorphism
$f$ from $S_{n}$ to the group $\mathrm{GL}_{t}(\mathbb{C})$ of
invertible $t\times t$ matrices, for some dimension $t\in\mathbb{N}$.
Examples include the trivial representation that maps all of $S_{n}$
to $1$, the sign representation that maps permutations to their sign,
and the permutation matrix representation that maps each permutation
to its $n\times n$ permutation matrix.

A representation $f:S_{n}\to\mathrm{GL}_{t}(\mathbb{C})$ is \textbf{irreducible
}if no proper subspace of $\mathbb{C}^{t}$ is invariant under all
the transformations $f(\pi)$ for $\pi\in S_{n}$. Among the examples
given before, this holds trivially for the trivial and sign representations.
The permutation matrix representation however is not irreducible for
$n>1$, as every permutation matrix maps the $1$-dimensional subspace
of $\mathbb{C}^{n}$ spanned by $(1,\ldots,1)$ to itself.

\subsection{Characters}

Characters condense essential information about representations $f:S_{n}\to\mathrm{GL}_{t}(\mathbb{C})$
into scalar-valued functions $\chi_{f}:S_{n}\to\mathbb{C}$. For us,
they play the role of ``generalized signs'' in the sum-product definition
of immanants.
\begin{defn}
The \textbf{character} of a representation $f$ is the function $\chi_{f}:S_{n}\to\mathbb{C}$
that maps $\pi\in S_{n}$ to the trace of the matrix $f(\pi)$.
\end{defn}

The trivial and sign representations coincide trivially with their
characters. The character of the permutation matrix representation
counts the fixed points of a permutation.

Characters of representations are \textbf{class functions}, which
are functions $f:S_{n}\to\mathbb{C}$ that depend only on the cycle
format of the input. These functions form a vector space by point-wise
linear combinations, and a particularly useful basis for this space
is given by the \textbf{irreducible characters}, which are the characters
of irreducible representations. The set of irreducible characters
corresponds bijectively to the partitions of $n$.

For $S_{2}$, the only irreducible characters are the trivial character
$\chi_{(2)}$ and sign character $\chi_{(1,1)}$. There are five irreducible
characters for $S_{4}$; their values $\chi_{\lambda}(\rho)$ are
shown below as a character table.
\begin{center}
\begin{tabular}{cc|c|c|c|c|c|}
 & \multicolumn{1}{c}{} & \multicolumn{5}{c}{$\rho$}\tabularnewline
 & \multicolumn{1}{c}{} & \multicolumn{5}{c}{\vspace{-5bp}
}\tabularnewline
 &  & $(1^{4})$ & $(2^{1},1^{2})$ & $(2^{2})$ & $(3^{1},1^{1})$ & $(4^{1})$\tabularnewline
\cline{2-7} \cline{3-7} \cline{4-7} \cline{5-7} \cline{6-7} \cline{7-7} 
\multirow{5}{*}{$\lambda\quad$} & $(4^{1})$ & $1$ & $1$ & $1$ & $1$ & $1$\tabularnewline
\cline{2-7} \cline{3-7} \cline{4-7} \cline{5-7} \cline{6-7} \cline{7-7} 
 & $(3^{1},1^{1})$ & $3$ & $1$ & $-1$ & $0$ & $-1$\tabularnewline
\cline{2-7} \cline{3-7} \cline{4-7} \cline{5-7} \cline{6-7} \cline{7-7} 
 & $(2^{2})$ & $2$ & $0$ & $2$ & $-1$ & $0$\tabularnewline
\cline{2-7} \cline{3-7} \cline{4-7} \cline{5-7} \cline{6-7} \cline{7-7} 
 & $(2^{1},1^{2})$ & $3$ & $-1$ & $-1$ & $0$ & $1$\tabularnewline
\cline{2-7} \cline{3-7} \cline{4-7} \cline{5-7} \cline{6-7} \cline{7-7} 
 & $(1^{4})$ & $1$ & $-1$ & $1$ & $1$ & $-1$\tabularnewline
\cline{2-7} \cline{3-7} \cline{4-7} \cline{5-7} \cline{6-7} \cline{7-7} 
 & \multicolumn{1}{c}{} & \multicolumn{1}{c}{} & \multicolumn{1}{c}{} & \multicolumn{1}{c}{} & \multicolumn{1}{c}{} & \multicolumn{1}{c}{}\tabularnewline
\end{tabular}
\par\end{center}

In the next subsections, we describe the Murnaghan-Nakayama rule,
a combinatorial method for calculating character values $\chi_{\lambda}(\rho)$,
together with a simple extension thereof that applies to particular
linear combinations of character values. To state these rules, we
need to introduce several types of \emph{tableaux}.

\subsection{Skew shapes and their tableaux}

Recall that partitions $\lambda$ can be described by Young diagrams,
such as

\begin{table}[h]
\centering{}%
\begin{tabular}{c}
\ydiagram{4,4,3,3,3,2}\tabularnewline
\tabularnewline
$(4^{2},3^{3},2^{1})$.\tabularnewline
\end{tabular}
\end{table}

Given such a diagram, a tableau is obtained by writing numbers (or
other objects) into the boxes, subject to some specified rules. The
representation theory of $S_{n}$ abounds in different types of tableaux---we
introduce yet another such type, the \emph{skew shape tableaux.}
\begin{defn}
Let $\lambda,\mu$ be partitions such that $\mu_{i}\leq\lambda_{i}$
for all rows $i$ of $\lambda$. We say that $\mu$ is \emph{contained}
in $\lambda$ and define the \textbf{skew shape} $\lambda/\mu$ by
removing the diagram of $\mu$ from $\lambda$.
\end{defn}

Consider the two examples below. The right example shows that skew
shapes need not be connected; we call a skew shape \textbf{connected}
if each pair of boxes can be reached by a path in the interior of
the shape.

\begin{table}[h]
\centering{}%
\begin{tabular}{ccc}
\ydiagram{2+2,2+2,1+2,1+2,3,2} & $\qquad$ & \ydiagram{3+1,3+1,3+0,1+2,1+2,2}\tabularnewline
 &  & \tabularnewline
$(4^{2},3^{3},2^{1})/(2^{2},1^{2})$ &  & $(4^{2},3^{3},2^{1})/(3^{3},1^{2})$\tabularnewline
\end{tabular}
\end{table}

Table~\ref{tab: conn-shapes} lists the connected skew shapes on
$4$ boxes. Any general skew shape on $b$ boxes is obtained by choosing
connected skew shapes with a total of $b$ boxes and arranging them
in disjoint rows and columns.
\begin{table}
\begin{centering}
\begin{tabular}{ccccccccc}
\ydiagram{1,1,1,1} & \ydiagram{2,1,1} & \ydiagram{1+1,1+1,2} & \ydiagram{2,2} & \ydiagram{1+1,2,1} & \ydiagram{3,1} & \ydiagram{2+1,3} & \ydiagram{1+2,2} & \ydiagram{4}\tabularnewline
\end{tabular}
\par\end{centering}
\caption{\label{tab: conn-shapes}All connected skew shapes on $4$ boxes.}
\end{table}

We will often peel skew shapes $\gamma$ from other skew shapes $\lambda/\mu$.
Our definition of this process allows for peeling different components
of $\gamma$ from different places.
\begin{defn}
A skew shape $\gamma$ can be \textbf{peeled} from $\lambda/\mu$
if there is a partition $\lambda'$ contained between $\mu$ and $\lambda$
such that $\lambda/\lambda'$ equals $\gamma$ after deleting empty
rows and columns from both $\lambda/\lambda'$ and $\gamma$.
\end{defn}

A skew shape \emph{tableau} is obtained by successively peeling skew
shapes from $\lambda$.
\begin{defn}
Let $\tilde{\Gamma}=(\Gamma_{1},\ldots,\Gamma_{s})$ be such that
$\Gamma_{i}$ for $i\in[s]$ is a set of skew shapes on the same number
$n_{i}$ of boxes. Let $\lambda$ be a partition of $n=\sum_{i}n_{i}$.
A \textbf{skew shape tableau} of $\lambda$ with format $\tilde{\Gamma}$
is obtained by successively peeling skew shapes $\gamma_{1},\ldots,\gamma_{s}$
with $\gamma_{i}\in\Gamma_{i}$ from $\lambda$ and labeling $\gamma_{i}$
with $i$ in $\lambda$. We write $\mathcal{S}(\lambda,\tilde{\Gamma})$
for the set of such tableaux.
\end{defn}

Three skew shape tableaux with different formats are shown below.
The boxes are \emph{colored} rather than numbered. The third tableau
is even a \emph{border strip} tableau, as defined in the next subsection.
\begin{table}[h]
\centering{}%
\begin{tabular}{ccccc}
\SSTabOne & $\ $ & \SSTabTwo & $\ $ & \SSTabThree\tabularnewline
\end{tabular}
\end{table}

\subsection{\label{subsec: The-Murnaghan-Nakayama-rule}An extension of the Murnaghan-Nakayama
rule}

The Murnaghan-Nakayama rule expresses the character value $\chi_{\lambda}(\rho)$
for partitions $\lambda$ and $\rho$ as a signed sum over particular
skew shape tableaux of $\lambda$. The sign of a tableau is determined
by the parity of odd-height shapes. In our later arguments, $\rho$
will always be the format of a cycle cover. 
\begin{defn}
A \textbf{border strip} is a connected skew shape not containing any
$2\times2$ square. A \textbf{border strip tableau} of $\lambda\vdash n$
is a skew shape tableau consisting only of border strips. Given a
partition $\lambda$ and an ordered partition $\kappa$, we write
$\mathcal{B}(\lambda,\kappa)$ for the set of border strip tableaux
of $\lambda$ in which the $i$-th shape has $\kappa_{i}$ boxes.

The \textbf{height} $\height(\gamma)$ of a border strip $\gamma$
is the number of occupied rows in $\gamma$ minus $1$. Given a skew
shape tableau $T$ with skew shapes $\gamma_{1},\ldots,\gamma_{s}$,
the \textbf{height sign} $\height(T)$ is defined as $\height(T)=\prod_{i=1}^{s}(-1)^{\height(\gamma_{i})}.$
\end{defn}

For the border strip tableau in the above example, the heights are
$5$ (green), $1$ (yellow), $2$ (red), and $2$ (blue), resulting
in an overall height sign of $+1$. Having established these preliminaries,
we can state the Murnaghan-Nakayama rule. For proofs, we refer to
textbooks~\cite{fulton1991representation,lorenz2018tour}.
\begin{thm}[Murnaghan-Nakayama rule]
\label{thm: mn-rule}For partitions $\lambda$ and $\rho$ and any
ordering $\kappa$ of $\rho$, we have 
\[
\chi_{\lambda}(\rho)=\sum_{T\in\mathcal{B}(\lambda,\kappa)}\height(T).
\]
\end{thm}

\begin{rem}
\label{rem: mn-rule-ordering}When invoking this rule, it makes sense
to choose a useful ordering $\kappa$ of $\rho$. For example, to
show $\chi_{(5,4,3,2,1)}(1^{3},2^{6})=0$, we can
\begin{itemize}
\item sum over a rather large number of border strip tableaux to observe
that their signs cancel, or alternatively
\item reorder $\rho=(1^{3},2^{6})$ to $\kappa=(2^{6},1^{3})$ and realize
directly that no border strip on $2$ boxes can be peeled from the
border of the staircase $(5,4,3,2,1)$.
\end{itemize}
More generally, this observation shows that $\chi_{\mu}(\rho)=0$
whenever $\mu$ is a staircase and $\rho$ is a partition that contains
at least one even part.
\end{rem}

The above remark will be crucial for the staircase-based reduction
in Section~\ref{sec: staircase}. For the tetromino-based reduction
in Section~\ref{sec: tetrominos}, we extend Theorem~\ref{thm: mn-rule}
to character evaluations on products of partition sets.
\begin{defn}
Let $F_{1},\ldots,F_{t}$ be sets of partitions such that each set
$F_{i}$ collects partitions of the same integer $d_{i}$. The\emph{
partition product }$F_{1}\times\ldots\times F_{t}$ is the multi-set
consisting of the $\prod_{i}|F_{i}|$ partitions of $d_{1}+\ldots+d_{t}$
obtained by choosing one partition from each set $F_{i}$ and concatenating
those $t$ partitions.
\end{defn}

We linearly extend class functions $f$ to multi-sets $S$ of partitions
by declaring $f(S)=\sum_{\rho\in S}f(\rho)$ and show how to calculate
$\chi_{\lambda}(F_{1}\times\ldots\times F_{t})$ combinatorially by
an extension of the Murnaghan-Nakayama rule. To this end, we define
admissible skew shapes for each $F_{i}$, each with a particular coefficient.
In analogy with the original rule, admissible skew shapes play the
role of border strips, and the coefficients of admissible skew shapes
play the role of heights of border strips.
\begin{defn}
\label{def: Gamma_F}Given a set of partitions $F$, let $\Gamma_{F}$
be the set of all skew shapes $\gamma$ that admit a border strip
tableau $T\in\mathcal{B}(\gamma,\rho)$ for some $\rho\in F$. For
$\gamma\in\Gamma_{F}$, we define the coefficient
\[
\alpha_{F}(\gamma)=\sum_{\rho\in F}\sum_{T\in\mathcal{B}(\gamma,\rho)}\height(T).
\]
\end{defn}

In preparation of Section~\ref{sec: tetrominos}, we exemplify this
definition with the cycle formats of active edge gadgets. For $F=\{(2^{2}),(4)\}$,
we observe that the set $\Gamma_{F}$ consists of all skew shapes
that can be covered with two disjoint dominos. Note that there are
shapes $\gamma\in\Gamma_{F}$ with $\alpha_{F}(\gamma)=0$: For example,
the vertical $4$-box line admits a border strip tableau for $(2^{2})$
and one for $(4)$, shown below. These tableaux have opposite height
signs and thus cancel.
\begin{table}[H]
\centering{}%
\begin{tabular}{ccccc}
\SSLineTwo &  & \SSLineOne &  & \SSLineEmpty\tabularnewline
\vspace{-5bp}
 &  &  &  & \tabularnewline
$\height(T)=1$ &  & $\height(T')=-1$ & $\quad\Rightarrow\quad$ & $\alpha_{F}(\gamma)=0$\tabularnewline
\end{tabular}
\end{table}
This outcome is expected: In Section~\ref{sec: tetrominos}, we show
that $4$-box shapes $\gamma$ with $\alpha_{F}(\gamma)\neq0$ serve
as a resource for establishing hardness of immanants. If we can peel
many such shapes from a partition $\lambda$, then we can reduce a
large permanent to the $\lambda$-immanant. Thus, if the vertical
$4$-box line $\gamma$ satisfied $\alpha_{F}(\gamma)\neq0$, our
reductions would allow us to establish hardness of the determinant.

With all relevant notions introduced, we can now turn to our generalization
of the Murnaghan-Nakayama rule. The proof is almost syntactic; it
essentially requires us to group border strips into skew shapes and
collect terms accordingly. Nevertheless, a little care is required
to ensure that the grouped border strips are indeed proper skew shapes.
\begin{lem}
\label{lem: mn-rule extended}Let $\lambda$ be a partition. Given
a partition product $F_{1}\times\ldots\times F_{t}$, abbreviate $\Gamma_{i}=\Gamma_{F_{i}}$
and $\alpha_{i}=\alpha_{F_{i}}$ for $i\in[t]$. Writing $\tilde{\Gamma}=(\Gamma_{1},\ldots,\Gamma_{t})$,
we have
\begin{equation}
\chi_{\lambda}(F_{1}\times\ldots\times F_{t})=\sum_{\substack{S\in\mathcal{S}(\lambda,\tilde{\Gamma})\\
\mathrm{with\ shapes\ }\gamma_{1}\ldots\gamma_{t}
}
}\prod_{i=1}^{t}\alpha_{i}(\gamma_{i}).\label{eq: skewproof-main}
\end{equation}
\end{lem}

\begin{proof}
Write $k_{i}:=|F_{i}|$ and enumerate $F_{i}=\{\rho_{i,1},\ldots,\rho_{i,k_{i}}\}$
for $i\in[t]$. We write $A=[k_{1}]\times\ldots\times[k_{t}]$ for
the set of multi-indices into $F_{1}\times\ldots\times F_{t}$. For
any multi-index $a\in A$, we define the ordered partition 
\[
\rho_{a}:=(\rho_{1,a(1)},\ldots,\rho_{t,a(t)})
\]
by concatenating the partitions $\rho_{i,a(i)}$ for $i\in[t]$. Here,
the order within $\rho_{i,a(i)}$ is not relevant, but the partitions
need to be concatenated in the specified order. We then have
\begin{align}
\chi_{\lambda}(F_{1}\times\ldots\times F_{t}) & =\sum_{a\in A}\sum_{T\in\mathcal{B}(\lambda,\rho_{a})}\height(T).\label{eq: skewproof-1}
\end{align}

Given $a\in A$ and a border strip tableau $T\in\mathcal{B}(\lambda,\rho_{a})$,
we define the skew shape tableau $S(T)$ by grouping, for each $i\in[t]$,
the different border strips corresponding to block $\rho_{i,a(i)}$
into a skew shape. Since these border strips are peeled consecutively
in $T$, the tableau $S(T)$ is indeed a skew shape tableau, and it
is contained in $\mathcal{S}(\lambda,\tilde{\Gamma})$ by definition
of the sets $\Gamma_{1},\ldots,\Gamma_{t}$. We then partition the
set of tableaux $T\in\mathcal{B}(\lambda,\rho_{a})$ according to
$S(T)$ and obtain 
\begin{equation}
\sum_{a\in A}\sum_{T\in\mathcal{B}(\lambda,\rho_{a})}\height(T)=\sum_{a\in A}\sum_{S\in\mathcal{S}(\lambda,\tilde{\Gamma})}\sum_{\substack{T\in\mathcal{B}(\lambda,\rho_{a})\\
\text{with }S(T)=S
}
}\height(T).\label{eq: skewproof-2}
\end{equation}
Let us change the order of summation and sum over $S\in\mathcal{S}(\lambda,\tilde{\Gamma})$
first. For a fixed skew shape tableau $S$ with shapes $\gamma_{1},\ldots,\gamma_{t}$,
each border strip tableau $T\in\bigcup_{a\in A}\mathcal{B}(\lambda,\rho_{a})$
with $S(T)=S$ is obtained by choosing a border strip tableau independently
for each of the shapes $\gamma_{1},\ldots,\gamma_{t}$. We obtain
\begin{align*}
\sum_{a\in A}\sum_{\substack{T\in\mathcal{B}(\lambda,\rho_{a})\\
\text{with }S(T)=S
}
}\height(T) & \ =\ \prod_{i=1}^{t}\sum_{j=1}^{k_{i}}\sum_{T\in\mathcal{B}(\gamma_{i},\rho_{i,j})}\height(T)\ =\ \prod_{i=1}^{t}\alpha_{i}(\gamma_{i}).
\end{align*}
By summing both sides over all $S\in\mathcal{S}(\lambda,\tilde{\Gamma})$
and combining the result with (\ref{eq: skewproof-2}) and (\ref{eq: skewproof-1}),
the lemma follows.
\end{proof}

\section{Staircases versus non-vanishing tetrominos\label{sec:Staircases-versus-skew}}

We associate various quantities with partitions $\lambda\vdash n$
to measure to what extent the $\lambda$-immanant lends itself to
a reduction from counting matchings. Then we investigate their interplay
with the number $b(\lambda)$ of boxes to the right of the first column
of $\lambda$.

The diagrams $\tHDom$ and $\tVDom$ of $(2^{1})$ and $(1^{2})$
will be denoted as \emph{domino}s. The domino number $d(\lambda)$
is the maximum number of dominos that can be peeled successively from
$\lambda$. As we show below, the result of peeling these dominos
from $\lambda$ is a partition of the form $\mu=(k,\ldots,1)$ for
some $k\in\mathbb{N}$. Such partitions and their associated shapes
are called \emph{staircases}, and we define $z(\lambda)=k+\ldots+1$
as the staircase size and $w(\lambda)=k$ as the staircase width of
$\lambda$. Note that $2d(\lambda)+z(\lambda)=n$ and that $z(\lambda)$
can be zero; this happens when $\lambda$ can be covered fully by
dominos.
\begin{fact}
\label{fact: extract-staircase}The shape $\mu$ obtained by peeling
$d(\lambda)$ dominos from $\lambda$ is the unique staircase on $z(\lambda)$
boxes.
\end{fact}

\begin{proof}
By maximality of $d(\lambda)$, no domino can be peeled from $\mu$,
so all gaps between consecutive rows in $\mu$ are $1$. This requires
$\mu$ to be a staircase. It has $n-2d(\lambda)=z(\lambda)$ boxes,
and the number of boxes uniquely determines a staircase.
\end{proof}
We define $s(\lambda)$ to be the maximum number of \emph{non-vanishing
tetrominos} (depicted on page~\pageref{fig: tetrofig} and Table~\ref{tab: nv-tetro})
that can be peeled successively from $\lambda$. As we establish in
Sections~\ref{sec: staircase} and \ref{sec: tetrominos}, any partition
$\lambda$ with large $w(\lambda)$ or $s(\lambda)$ induces an immanant
to which the hard problem of counting $k$-matchings can be reduced.
In this subsection, we show that at least one of these numbers is
large if $b(\lambda)$ is large. Towards this, we first observe that
partitions $\lambda$ with $s(\lambda)=0$ are strongly restricted. 
\begin{lem}
\label{lem: mine-gaps}For any partition $\lambda$ with $s(\lambda)=0$,
at least one of the following holds:
\begin{enumerate}
\item The shape $\lambda$ is a staircase up to some number of additional
boxes in the first column.
\item There is an index $i^{*}$ such that, after peeling a horizontal domino
from each row $i\leq i^{*}$, the remaining shape is the staircase
of $\lambda$.
\end{enumerate}
\end{lem}

\begin{proof}
There is at most one row $i$ with gap $\delta_{i}>1$: If $\delta_{i},\delta_{i'}\geq2$
for $i\neq i'$, then we could peel $\tHDom$ from these two rows.
As these dominos share no rows or columns, they constitute a non-vanishing
tetromino, contradicting $s(\lambda)=0$. Similarly, we see that the
set of rows $i$ with $\delta_{i}=0$, if non-empty, forms a contiguous
interval $[j,k]$: If the set contained two disjoint intervals, we
could peel two $\tVDom$ from their ends.

In the following, we distinguish whether $\lambda$ contains a row
with gap $0$ or not. If such a row exists, we show that the first
case of the lemma statement applies. Otherwise, we establish the second
case.
\begin{enumerate}
\item If there is a row $i$ with gap $\delta_{i}=0$, let $[j,k]\ni i$
be the single interval of rows with gap $0$. We rule out the existence
of a row $i'$ with $\delta_{i'}>1$, as it would lead to the following
contradictions:
\begin{itemize}
\item If $i'<j$ or $i'>k+1$, then we could peel $\tHDom$ from row $i'$
and $\tVDom$ from rows $k-1$ and $k$. This contradicts $s(\lambda)=0$.
\item If $i'=k+1$, then we could peel $\tSquare$ from rows $k$ and $k+1$.
\end{itemize}
It follows that $\delta_{i'}=1$ for all $i'\notin[j,k]$. Next, we
show that $\lambda_{i}=1$ for all $i\in[j,k]$: If $\lambda_{j}=\lambda_{k}>1$,
then $\tSnake$ could be peeled from rows $k$ to $k+2$ for a contradiction.
(Indeed, we would have $\delta_{k+1}=1$, since $\delta_{i'}=1$ for
all rows $i'\notin[j,k]$, and $\lambda_{k}>1$ would that row $k+2$
exists.) Hence, the first case of the lemma applies.
\item Otherwise, we have $\delta_{i}>0$ for all rows $i$. As discussed
in the first paragraph of the proof, there is at most one row $i^{*}$
with $\delta_{i^{*}}>1$.
\begin{itemize}
\item If $\delta_{i^{*}}\geq4$, we could peel $\tLine$ from row $i^{*}$,
contradicting $s(\lambda)=0$.
\item If $\delta_{i^{*}}=2$ and row $i^{*}+1$ exists, we could peel $\tL$
from rows $i^{*}$ and $i^{*}+1$.
\end{itemize}
It follows that $\delta_{i^{*}}=3$, or that $i^{*}$ is the last
row and $\delta_{i^{*}}=2$. As $i^{*}$ is the only row with gap
$\delta_{i^{*}}\neq1$, the second case of the lemma applies.

\end{enumerate}
Since at least one of the two above cases applies, the lemma is proven.
\end{proof}
It follows that partitions $\lambda$ with $s(\lambda)=0$ contain
only few boxes outside their first column and staircase.
\begin{cor}
\label{cor: mine-tetro}Let $\lambda$ be a partition with $s(\lambda)=0$
and staircase $\mu$ of width $w\in\mathbb{N}$. Then $b(\lambda/\mu)\leq2w+1$.
\end{cor}

\begin{proof}
If the second case of Lemma~\ref{lem: mine-gaps} applies to $\lambda$,
then $b(\lambda/\mu)\leq2w+1$. If the first case applies, let $\mu'$
denote the largest staircase contained in $\lambda$. (This need not
be the staircase of $\lambda$.) Then $\lambda$ is $\mu'$ with some
number $a\in\mathbb{N}$ of additional boxes in the first column.
When peeling dominos from $\lambda$, the first dominos must be contained
in the first column.
\begin{itemize}
\item If $a$ is even, then $\mu'$ is the staircase of $\lambda$, so $b(\lambda/\mu)=0$:
Peeling $\tVDom$ from the first column results in $\mu'$.
\item If $a$ is odd, then the lowest box of $\mu'$ is contained in the
last $\tVDom$ that can be peeled from the first column. We can then
peel a horizontal domino from each of the remaining $w+1$ rows. The
remainder is the staircase of $\lambda$. It follows that $b(\lambda/\mu)\leq2w+1$.
\end{itemize}
In both cases, we obtain $b(\lambda/\mu)\leq2w+1$, thus proving the
corollary.
\end{proof}
We are ready to prove the main lemma of this section.
\begin{lem}
\label{lem: mine-main}For any partition $\lambda$, at least one
of $s(\lambda)\geq b(\lambda)/8$ or $w(\lambda)\geq\sqrt{b}-1$ holds.
\end{lem}

\begin{proof}
Let $\mu$ be the staircase of $\lambda$ and $w=w(\lambda)$. The
number of boxes to the right of the first column in $\lambda/\mu$
is $b(\lambda/\mu)=b(\lambda)-b(\mu)=b(\lambda)-{w \choose 2}.$ By
applying the contraposition of Corollary~\ref{cor: mine-tetro} repeatedly,
it follows that
\[
s(\lambda)\geq\frac{b(\lambda/\mu)-(2w+1)}{4}=\frac{b(\lambda)-{w \choose 2}-(2w+1)}{4}=\frac{b(\lambda)}{4}-\frac{w^{2}+3w+2}{8}.
\]
We obtain
\[
s(\lambda)+\frac{w^{2}+3w+2}{8}\geq\frac{b(\lambda)}{4},
\]
and hence, at least one of the two terms is larger than $b(\lambda)/8$.
This implies the lemma.
\end{proof}

\section{Exploiting a staircase\label{sec: staircase}}

We show how to count $k$-matchings in $n$-vertex graphs $H$ with
access to the $\lambda$-immanant for a partition $\lambda$ with
staircase size $z(\lambda)\in\Omega(k)$. If $k\ll n$, we also require
a large domino number $d(\lambda)$. Our reduction relies on the intermediate
problem of counting cycle covers with a particular ``onion'' format,
outlined in Section~\ref{subsec: onion}, which we then reduce to
the $\lambda$-immanant using a graph construction that has a favorable
interplay with staircase characters, as established in Section~\ref{subsec: staircase-characters}.
We collect the steps in Section~\ref{subsec: staircase-hardness}.

\subsection{\label{subsec: onion}Onion partitions}

We first describe a particular ``onion partition'' derived from
$\lambda$, depicted in Figure~\ref{fig: onion-part}. To define
this partition, recall that $w(\lambda)$ and $z(\lambda)$ denote
the staircase width and size of $\lambda$, respectively.
\begin{defn}
\label{def: onion-cycle-formats}Let $\lambda\vdash n'$ be a partition
with staircase $\mu$. For $\ell\leq w(\lambda)/2$, let $\theta$
be the partition obtained by peeling $\ell$ maximal-length border
strips from $\mu$ and recording their lengths. We write $\left\Vert \theta\right\Vert $
for the number of boxes in $\theta$. Then the \emph{$\ell$-layer
onion} $\rho$ of $\lambda$ is the partition $(2^{d(\lambda)},\theta,1^{z(\lambda)-\left\Vert \theta\right\Vert })$.
For $k\in\mathbb{N}$, we say that $\rho$\emph{ accommodates} $k$
\emph{edges} if $\left\Vert \theta\right\Vert \geq2k+\ell$.
\begin{figure}
\begin{centering}
\ydiagram
[*(cyan2)]{6+2,5+2,4+2,3+2,2+2,1+2,2,1}
*[*(cyan1)]{8+2,7+2,6+2,5+2,4+2,3+2,2+2,1+2,2,1}
*[*(orange)]{10,9,8,7,6,5,4,3,2,1}
*[*(gray)]{14,13,12,9,8,5,4,3,2,1}
\par\end{centering}
\caption{\label{fig: onion-part}A partition $\lambda$ with $d(\lambda)=8$
dominos, staircase width $w(\lambda)=10$ and size $z(\lambda)=55$.
The $2$-layer onion $\rho=({\color{gray} 2^8},{\color{cyan1} 19},{\color{cyan2} 15},{\color{orange} 1^{21}})$
of $\lambda$ is shown as a border strip tableau; parts corresponding
to dominos and singletons are not distinguished in this figure. All
border strip tableaux for $\rho$ in $\lambda$ have the depicted
form. This onion accommodates $16$ edges; the two layers contain
$34$ boxes, and one box per layer is required for \textquotedblleft closing\textquotedblright{}
the layer.}
\end{figure}
\end{defn}

In words, the $\ell$-layer onion $\rho$ of $\lambda$ is obtained
as follows: First peel all $d(\lambda)$ dominos from $\lambda$ to
expose the staircase $\mu$. Then peel $\ell$ maximal-length border
strips from $\mu$. Finally, peel the remaining $z(\lambda)-\left\Vert \theta\right\Vert $
boxes of $\mu$ as singletons. If $\rho$ accommodates $k$ edges,
then the border strips contain $2k$ boxes and one extra box for each
peeled border strip; such partitions allow us to reduce counting $k$-matchings
to counting $\rho$-cycle covers.

Next, we describe two related ways of constructing a digraph $G$
whose cycle covers of onion format correspond to $k$-matchings in
$H$. Given an edge-weighted digraph $G$, let 
\[
\CC(G,\rho)=\sum_{\substack{\rho\text{-cycle cover }\\
C\,\text{in }G
}
}\prod_{e\in C}w(e).
\]
Furthermore, we write $\coeff{x^{s}}p$ for the coefficient of $x^{s}$
in a polynomial $p\in\mathbb{Q}[x]$. We also say that a vertex-set
$T$ in a graph $G$ is an \emph{odd-cycle transversal} if $G-T$
is bipartite up to self-loops. (Note that not requiring self-loops
to be removed by $T$ may not be the standard way of defining odd-cycle
transversals.)
\begin{lem}
\label{lem: staircase-graph}Let $H$ be a bipartite $n$-vertex graph
and let $\lambda$ be an integer partition.
\begin{enumerate}
\item If there are onions of $\lambda$ that accommodate $n/2$ edges, let
$\rho$ be such an onion with the minimal number of layers. In polynomial
time, we can construct a digraph $G$ with $\CC(G,\rho)=(n/2)!\cdot\PerfMatch(H).$
\item For any $k\in\mathbb{N}$ with $d(\lambda)\geq n-2k$: If there are
onions of $\lambda$ that accommodate $k$ edges and contain $\geq2k$
copies of $1$, let $\rho$ be such an onion with the minimal number
of layers. In polynomial time, we can construct a digraph $G$ (whose
edge-weights may be constant multiples of an indeterminate $x$) such
that $\coeff{x^{2k}}{\CC(G,\rho)}=k!\cdot\match(H,k)$.
\end{enumerate}
In both cases, let $\rho=(2^{d(\lambda)},\theta,1^{z(\lambda)-\left\Vert \theta\right\Vert })$
be the relevant $\ell$-layer onion, for $\ell\in\mathbb{N}$. Then
the following holds:
\begin{itemize}
\item $G$ admits an odd-cycle transversal $T\subseteq V(G)$ of cardinality
$\ell$.
\item Every cycle cover in $G$ (in the second case, every cycle cover whose
weight is a constant multiple of $x^{2k}$) contains exactly $z(\lambda)-\left\Vert \theta\right\Vert $
self-loops and at least $d(\lambda)$ digons.
\end{itemize}
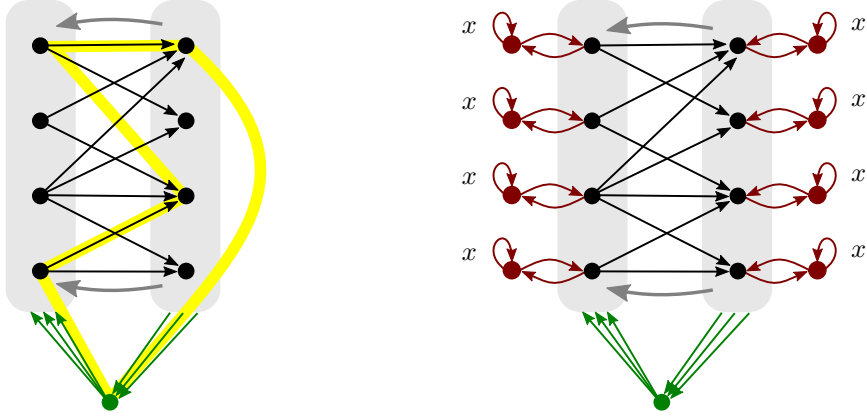
\begin{figure}
\centering \def\svgwidth{12cm} 
\begingroup%
  \makeatletter%
  \providecommand\color[2][]{%
    \errmessage{(Inkscape) Color is used for the text in Inkscape, but the package 'color.sty' is not loaded}%
    \renewcommand\color[2][]{}%
  }%
  \providecommand\transparent[1]{%
    \errmessage{(Inkscape) Transparency is used (non-zero) for the text in Inkscape, but the package 'transparent.sty' is not loaded}%
    \renewcommand\transparent[1]{}%
  }%
  \providecommand\rotatebox[2]{#2}%
  \newcommand*\fsize{\dimexpr\f@size pt\relax}%
  \newcommand*\lineheight[1]{\fontsize{\fsize}{#1\fsize}\selectfont}%
  \ifx\svgwidth\undefined%
    \setlength{\unitlength}{223.44392071bp}%
    \ifx\svgscale\undefined%
      \relax%
    \else%
      \setlength{\unitlength}{\unitlength * \real{\svgscale}}%
    \fi%
  \else%
    \setlength{\unitlength}{\svgwidth}%
  \fi%
  \global\let\svgwidth\undefined%
  \global\let\svgscale\undefined%
  \makeatother%
  \begin{picture}(1,0.4552922)%
    \lineheight{1}%
    \setlength\tabcolsep{0pt}%
    \put(0,0){\includegraphics[width=\unitlength,page=1]{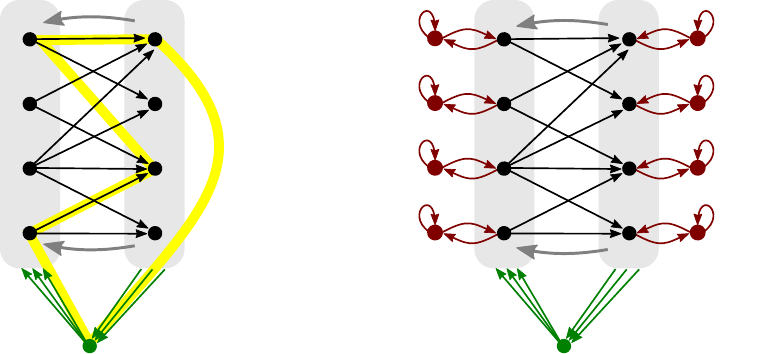}}%
    \put(0.50528332,0.41830651){\makebox(0,0)[lt]{\lineheight{1.25}\smash{\begin{tabular}[t]{l}$x$\end{tabular}}}}%
    \put(0.50528332,0.33465525){\makebox(0,0)[lt]{\lineheight{1.25}\smash{\begin{tabular}[t]{l}$x$\end{tabular}}}}%
    \put(0.50528332,0.25100395){\makebox(0,0)[lt]{\lineheight{1.25}\smash{\begin{tabular}[t]{l}$x$\end{tabular}}}}%
    \put(0.50528332,0.16735266){\makebox(0,0)[lt]{\lineheight{1.25}\smash{\begin{tabular}[t]{l}$x$\end{tabular}}}}%
    \put(0.93663202,0.42162898){\makebox(0,0)[lt]{\lineheight{1.25}\smash{\begin{tabular}[t]{l}$x$\end{tabular}}}}%
    \put(0.93663202,0.33797772){\makebox(0,0)[lt]{\lineheight{1.25}\smash{\begin{tabular}[t]{l}$x$\end{tabular}}}}%
    \put(0.93663202,0.25432642){\makebox(0,0)[lt]{\lineheight{1.25}\smash{\begin{tabular}[t]{l}$x$\end{tabular}}}}%
    \put(0.93663202,0.17067513){\makebox(0,0)[lt]{\lineheight{1.25}\smash{\begin{tabular}[t]{l}$x$\end{tabular}}}}%
  \end{picture}%
\endgroup%

\caption{The graphs from Lemma~\ref{lem: staircase-graph}. To reduce clutter,
edges from $R$ to $L$ are only hinted, and only one transit vertex
is shown. In the left graph, a cycle corresponding to a $2$-matching
is displayed in yellow color.}
\end{figure}
\end{lem}

\begin{proof}
Let $V(H)=L\cup R$ and write $\rho=(2^{d(\lambda)},\theta,1^{z(\lambda)-\left\Vert \theta\right\Vert })$
for the relevant $\ell$-layer onion of $\lambda$. By adding isolated
edges to $H$, we assume that $n/2$ (in the first case) or $k$ (in
the second case) is the maximum number of edges accommodated by $\lambda$.
For the \textbf{first part} of the lemma, the graph $G$ is defined
as follows:
\begin{enumerate}
\item Direct all edges in $H$ from $L$ to $R$ and add all edges $R\times L$.
\item Add a set of \emph{transit vertices} $T=\{t_{1},\ldots t_{\ell}\}$
and all edges in $R\times T$ and $T\times L$.
\item Add $d(\lambda)$ disjoint \emph{padding digons} and \emph{$z(\lambda)-\left\Vert \theta\right\Vert $
padding vertices} with self-loops.
\end{enumerate}
Any $\rho$-cycle cover of $G$ uses all padding digons and loops.
Any remaining odd-length cycle must use a vertex of $T$, as it would
otherwise be contained in a bipartite graph. This implies that $T$
is an odd-cycle transversal. It follows that the $\ell$ (odd-length)
parts in $\theta$ can only be accommodated by $\ell$ disjoint cycles
in $G$ that each include exactly one transit vertex. Hence, any $\rho$-cycle
cover $C$ in $G$ induces a perfect matching in $H$ when restricted
to edges from $L$ to $R$: Deleting transit vertices results in a
cover of $V(H)$ with paths of odd length $\geq1$ that start in $L$,
and deleting the edges from $R$ to $L$ then leaves us with a matching.

Conversely, every perfect matching in $H$ induces exactly $(n/2)!$
cycle covers of format $\rho$ in $G$: For $\theta=(2r_{1}+1,2r_{2}+1,\ldots)$,
there are ${n/2 \choose r_{1},\ldots,r_{\ell}}$ ways of partitioning
the $n/2$ matching edges into the cycles corresponding to transit
vertices $t_{1},\ldots t_{\ell}$, and there are $r_{i}!$ ways of
choosing an ordering of the edges and the $i$-th transit vertex within
the $i$-th cycle.\footnote{Each such ordering corresponds to a pair $(j,\sigma)$ for $j\in[r_{i}]$
and a cyclic permutation $\sigma$ of $[r_{j}]$. The index $j$ indicates
that the transit vertex $t_{i}$ is visited immediately after the
$j$-th edge, and $\sigma$ describes the order of the edges along
the cycle. There are $r_{i}\cdot(r_{i}-1)!=r_{i}!$ such pairs $(j,\sigma)$.} This concludes the first part of the lemma.

For the \textbf{second part}, we construct $G$ by performing steps
1 and 2 from above, followed by these steps:
\begin{enumerate}
\item[3.] For each vertex $v\in V(H)$, add a \emph{switch vertex} $s_{v}$
with a self-loop of weight $x$, and a \emph{switch digon} between
$v$ and $s_{v}$.
\item[4.] Add $d(\lambda)-(n-2k)$ \emph{padding digons} and \emph{$z(\lambda)-\left\Vert \theta\right\Vert -2k$
padding vertices} with self-loops.
\end{enumerate}
Again, any cycle cover $C$ of $G$ uses all padding elements. The
weight of $C$ is $x^{2k}$ iff it includes exactly $2k$ self-loops
at switch vertices; it then includes $n-2k$ switch digons touching
the remaining switch vertices, so it contains at least $d(\lambda)$
digons and $z(\lambda)-\left\Vert \theta\right\Vert $ self-loops.
As before, if $C$ has format $\rho$, then the $\ell$ remaining
odd-length cycles induce a $k$-matching in $H$ when restricted to
edges from $L$ to $R$, and any $k$-matching in $H$ can be extended
to $k!$ such cycle covers.
\end{proof}
The odd-cycle transversal of $G$ and special properties of staircase
characters ensure that we can determine $\CC(G,\rho)$ for an onion
$\rho$ of $\lambda$ by reduction to $\imm_{\lambda}(G)$. We prove
this in the next subsection.

\subsection{Staircase characters\label{subsec: staircase-characters}}

Given a partition $\lambda$ with staircase $\mu$, consider any partition
$\rho=(2^{d(\lambda)},\rho')$ obtained by peeling the maximum number
of dominos from $\lambda$, followed by some other partition $\rho'$
of the staircase $\mu$ thusly exposed. We show how to relate $\chi_{\lambda}(\rho)$
to $\chi_{\mu}(\rho')$ by means of \emph{domino tilings} of $\lambda/\mu$,
which are border strip tableaux that only contain the dominos $\tVDom$
and $\tHDom$.
\begin{defn}
A \emph{domino tiling }of a skew shape $\lambda/\mu$ is a border
strip tableau $T$ of $\lambda/\mu$ with format $(2^{\left\Vert \lambda/\mu\right\Vert /2})$.
The \emph{parity} of a domino tiling $T$ is the parity of the number
of $\tVDom$ in $T$.
\end{defn}

Note that the parity of a domino tiling $T$ is essentially its height
sign $\height(T)$: The parity of $T$ is even/odd iff the height
sign of $T$ is positive/negative. Curiously, these numbers do not
depend on $T$.
\begin{lem}
\label{lem: domino-tiling-sign}All domino tilings of a fixed shape
$\lambda/\mu$ have the same parity.
\end{lem}

\begin{proof}
Paint the rows of $\lambda/\mu$ black and white in an alternating
way. In any domino tiling $T$, every vertical (horizontal) domino
contains an odd (even) number of white boxes. Hence, the number of
vertical dominos in $T$ agrees in parity with the number of white
boxes in $\lambda/\mu$, which does not depend on $T$.
\end{proof}
This elementary yet crucial observation gives the desired connection
between $\chi_{\lambda}(\rho)$ and $\chi_{\mu}(\rho')$.
\begin{lem}
\label{lem: stair-tiling-quotient}Let $\lambda$ be a partition with
staircase $\mu$, and let $\rho=(2^{d(\lambda)},\rho')$ for a partition
$\rho'$. Then we have $\chi_{\lambda}(\rho)\neq0$ iff $\chi_{\mu}(\rho')\neq0$.
\end{lem}

\begin{proof}
As any way of peeling $d(\lambda)$ dominos from $\lambda$ results
in $\mu$, the Murnaghan-Nakayama rule (Theorem~\ref{thm: mn-rule})
shows that $\chi_{\lambda}(\rho)=c_{\lambda/\mu}\cdot\chi_{\mu}(\rho')$
with
\begin{equation}
c_{\lambda/\mu}=\sum_{\substack{\text{domino tiling}\\
T\text{ of }\lambda/\mu
}
}\height(T),\label{eq: def-domino-sum}
\end{equation}
By Lemma~\ref{lem: domino-tiling-sign}, each term in (\ref{eq: def-domino-sum})
has the same sign. Since $\lambda/\mu$ is obtained by peeling $d(\lambda)$
dominos from $\lambda$, there is at least one domino tiling, and
hence there is at least one term in the sum. It follows that $c_{\lambda/\mu}>0$,
thus proving the lemma.
\end{proof}
The last missing piece is to recall from Remark~\ref{rem: mn-rule-ordering}
that the staircase character $\chi_{\mu}(\rho')$ vanishes whenever
$\rho'$ contains an even part. This allows us to analyze cycle covers
in the particular graphs $G$ constructed in the last subsection:
Any such cycle cover (with the right number of self-loops and digons)
is counted by $\imm_{\lambda}(G)$ iff its format is an onion.
\begin{lem}
\label{lem: right-format}Let $\lambda\vdash n'$ be a partition and
let $\rho=(2^{d(\lambda)},\theta,1^{z(\lambda)-\left\Vert \theta\right\Vert })$
be the $\ell$-layer onion of $\lambda$. Let $G$ be an $n'$-vertex
digraph with an odd-cycle transversal of size $\ell$. For any cycle
cover $C$ in $G$ that contains at least $d(\lambda)$ digons and
exactly $z(\lambda)-\left\Vert \theta\right\Vert $ self-loops: If
$\chi_{\lambda}(C)\neq0$, then the format of $C$ is $\rho$.
\end{lem}

\begin{proof}
Let $C$ be a cycle cover in $G$ with format $\beta=(2^{d(\lambda)},\rho')$
and $\chi_{\lambda}(\beta)\neq0$. By Lemma~\ref{lem: stair-tiling-quotient},
we have $\chi_{\mu}(\rho')\neq0$, where $\mu$ is the staircase of
$\lambda$. Since $G$ has an odd-cycle transversal of size $\ell$,
there are at most $\ell$ non-singleton odd parts in $\rho'$, as
every non-singleton odd cycle uses exactly one transversal vertex.
Since $\rho'$ contains $z(\lambda)-\left\Vert \theta\right\Vert $
singletons, the non-singleton odd parts contain $\left\Vert \theta\right\Vert $
boxes in total. This can only be achieved by peeling $\ell$ maximal-length
border strips from $\mu$, so the cycle format of the non-singleton
parts is $\theta$. Thus, we have $\rho'=(\theta,1^{z(\lambda)-\left\Vert \theta\right\Vert })$.
\end{proof}
\begin{cor}
\label{cor: onion-formats}Let $\lambda\vdash n'$ and let $\rho=(2^{d(\lambda)},\theta,1^{z(\lambda)-\left\Vert \theta\right\Vert })$
be the $\ell$-layer onion of $\lambda$. Let $G$ be an $n'$-vertex
graph with an odd-cycle transversal of size $\ell$. For $t\in\mathbb{N}$,
if every cycle cover whose weight is a constant multiple of $x^{t}$
contains at least $d(\lambda)$ digons and $z(\lambda)-\left\Vert \theta\right\Vert $
self-loops, then 
\[
\coeff{x^{t}}{\imm_{\lambda}(G)}=\chi_{\lambda}(\rho)\cdot\coeff{x^{t}}{\CC(G,\rho)}
\]
with $\chi_{\lambda}(\rho)\neq0$.
\end{cor}

\begin{proof}
The equation follows directly from Lemma~\ref{lem: right-format}
and the requirement on $G$. To show $\chi_{\lambda}(\rho)\neq0$,
it suffices by Lemma~\ref{lem: stair-tiling-quotient} to show $\chi_{\mu}(\rho')\neq0$
for $\rho'=(\theta,1^{z(\lambda)-\left\Vert \theta\right\Vert })$.
There is exactly one way of peeling $\theta$ from $\rho'$, and each
of the $z(\lambda)-\left\Vert \theta\right\Vert $ singletons peeled
from the resulting staircase incurs positive height sign. It follows
that all border strip tableaux of format $\rho'$ in $\mu$ contribute
to $\chi_{\mu}(\rho')$ with the same sign.
\end{proof}
Note that we can consider $G$ to be unweighted and invoke the corollary
with $t=0$. This yields the unweighted setting (in $x$) from the
first part of Lemma~\ref{lem: staircase-graph}.

\subsection{\label{subsec: staircase-hardness}Reductions}

We combine the results from the previous sections into reductions
from counting matchings to evaluating immanants for partitions with
large staircases.
\begin{lem}
\label{lem: finalred-staircase}The following can be achieved in polynomial
time and with polynomial-sized arithmetic circuits: 
\begin{enumerate}
\item Given a bipartite $n$-vertex graph $H$ and a partition $\lambda$
with $w(\lambda)\geq2\sqrt{n}$, compute a digraph $G$ and a number
$c\in\mathbb{Q}$ such that $\PerfMatch(H)=c\cdot\imm_{\lambda}(G)$.
\item Given a bipartite $n$-vertex graph $H$ and $k\in\mathbb{N}$, as
well as a partition $\lambda$ with $w(\lambda)\geq4\sqrt{k}$ and
$d(\lambda)\geq n-2k$, compute a digraph $G$ and a number $c\in\mathbb{Q}$
such that $\match(H,k)=c\cdot\coeff{x^{k}}{\imm_{\lambda}(G)}$.
\end{enumerate}
\end{lem}

\begin{proof}
For the \textbf{first part}, note that $w(\lambda)\geq2\sqrt{n}$
implies that $\lambda$ admits a $t$-layer onion $\rho=(2^{d(\lambda)},\theta,1^{z(\lambda)-\left\Vert \theta\right\Vert })$
for $t=\lfloor\sqrt{n}\rfloor$. For large enough $n$, the number
of edges accommodated by this onion is at least 
\[
\frac{{2t+1 \choose 2}-t}{2}=t^{2}\geq n-2\sqrt{n}+1\geq n/2.
\]
Hence, Lemma~\ref{lem: staircase-graph} allows us to construct a
digraph $G$ with $\CC(G,\rho)=(n/2)!\cdot\PerfMatch(H)$ such that
$G$ admits an odd-cycle transversal of cardinality $t$, and every
cycle cover in $G$ contains $z(\lambda)-\left\Vert \theta\right\Vert $
self-loops and at least $d(\lambda)$ digons. By Corollary~\ref{cor: onion-formats},
we then have $\imm_{\lambda}(G)=\chi_{\lambda}(\rho)\cdot\CC(G,\rho)$
with $\chi_{\lambda}(\rho)\neq0$. It follows from the last two equations
that 
\[
\PerfMatch(H)=\frac{\imm_{\lambda}(G)}{(n/2)!\cdot\chi_{\lambda}(\rho)}.
\]
Note that $\chi_{\lambda}(\rho)$ can be computed in polynomial time
and with polynomial-sized arithmetic circuits, see Proposition~7.4
in \cite{DBLP:books/daglib/0025071}.

For the \textbf{second part}, note that $w(\lambda)\geq4\sqrt{k}$
implies that $\lambda$ admits an onion $\rho=(2^{d(\lambda)},\theta,1^{z(\lambda)-\left\Vert \theta\right\Vert })$
that accommodates $\geq k$ edges. Additionally, $w(\lambda)\geq4\sqrt{k}$
also implies that $z(\lambda)-\left\Vert \theta\right\Vert \geq2k$.
Since $d(\lambda)\geq n-2k$ holds by assumption, all conditions for
the second part of Lemma~\ref{lem: staircase-graph} are thus fulfilled,
and we can construct a graph $G$ such that $\coeff{x^{2k}}{\CC(G,\rho)}=k!\cdot\match(H,k)$
and any cycle cover of weight $x^{2k}$ in $x$ contains exactly $z(\lambda)-\left\Vert \theta\right\Vert $
self-loops and at least $d(\lambda)$ digons. Corollary~\ref{cor: onion-formats}
then implies that $\coeff{x^{2k}}{\imm_{\lambda}(G)}=\chi_{\lambda}(\rho)\cdot\coeff{x^{2k}}{\CC(G,\rho)}$
with $\chi_{\lambda}(\rho)\neq0$. We conclude that 
\[
\match(H,k)=\frac{\coeff{x^{2k}}{\imm_{\lambda}(G)}}{k!\cdot\chi_{\lambda}(\rho)}.
\]
This proves the lemma.
\end{proof}

\section{Exploiting non-vanishing tetrominos\label{sec: tetrominos}}

We show how to count $k$-matchings with access to $\lambda$-immanants
for partitions $\lambda$ with large non-vanishing tetromino number
$s(\lambda)$: In Section~\ref{subsec: tetro-construct}, we use
edge gadgets to construct particular immanants $\imm_{\lambda}(G)$
that count matchings up to a multiplicative constant factor. We then
show, in Section~\ref{subsec: tetro-char}, that the factor arising
in the above construction is non-zero. To this end, we prove that
the character $\chi_{\lambda}$ does not vanish on a particular partition
product. Finally, we combine these insights in Section~\ref{subsec: tetro-hardness}
to obtain a reduction from counting matchings.

\subsection{\label{subsec: tetro-construct}Main construction}

Consider the edge gadget $Q$ depicted on page~\pageref{fig: eq-gadget}.
Intuitively speaking, this graph fragment ensures that unwanted cycle
covers in $G$ annihilate in $\imm_{\lambda}(G)$. In the remaining
\emph{consistent} cycle covers $C$, the edges within each edge gadget
$Q$ cover either \emph{both} endpoints (in an active state, as listed
in Figure~\ref{fig: eq-states}) or \emph{none} of them (in a passive
state). In particular, this allows us to interpret an active gadget
$Q$ as a matching edge between its endpoints, since the active state
prevents any edge outside of $Q$ from being incident with the endpoints
of $Q$. A similar approach was taken by the author together with
Bläser~\cite{DBLP:conf/mfcs/BlaserC11} to establish hardness of
the so-called \emph{cover polynomial}.
\begin{lem}
\label{lem: gadget-consistent}Let $G$ be a directed graph containing
copies $Q_{1},\ldots,Q_{t}$ of the edge-gadget $Q$, with endpoints
$e_{i}=\{u_{i},v_{i}\}$ for $i\in[t]$, such that distinct edge-gadgets
intersect only at endpoints. Let $\mathcal{C}^{*}(G)$ denote the
set of \emph{consistent} cycle covers $C$ in $G$: In such cycle
covers, the restriction $C\cap E(Q_{i})$ for $i\in[t]$ is an active
or passive state, as depicted in Figure~\ref{fig: eq-states}. Then
we have
\[
\imm_{\lambda}^{*}(G):=\sum_{C\in\mathcal{C}^{*}(G)}\chi_{\lambda}(C)\prod_{e\in C}w(e)\ =\ \imm_{\lambda}(G).
\]
\end{lem}

\begin{proof}
For $i\in[t]$, we say that edges from $E(G)\setminus E(Q_{i})$ are
$i$-external. For any cycle cover $C\in\mathcal{C}(G)$, the total
number of $i$-external incoming edges at the endpoints $u_{i}$ and
$v_{i}$ of $Q_{i}$ equals the total number of $i$-external outgoing
edges at $u_{i}$ and $v_{i}$: Otherwise, the vertices of $Q_{i}$
cannot be covered by cycles. This gives six possible in/out-degree
combinations at $u_{i}$ and $v_{i}$. 
\begin{figure}
\centering \def\svgwidth{16cm} 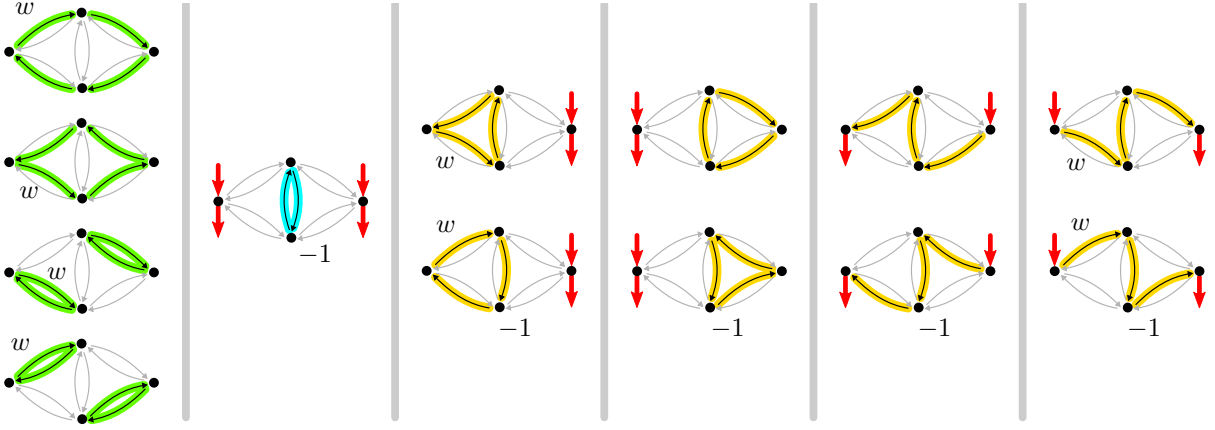

\caption{\label{fig: eq-states}Each column lists the edge-sets $C_{i}$ that
can arise for fixed combinations of endpoint in/out-degrees at gadget
$Q_{i}$ in $E(G)\setminus E(Q_{i})$. The first column lists \emph{active}
states, the second column shows the \emph{passive} state. All other
states come with annihilating partners and cancel out, as shown in
Lemma~\ref{lem: gadget-consistent}.}
\end{figure}
 For each of these combinations, the possible restrictions $C\cap E(Q_{i})$
are listed in Figure~\ref{fig: eq-states}.

Using this list, we show that the set $\mathcal{C}\setminus\mathcal{C}^{*}$
of non-consistent cycle covers can be partitioned into pairs $\{C,\overline{C}\}$
whose constituents have equal cycle formats $\rho(C)=\rho(\overline{C})$
and opposite weights $w(C)=-w(\overline{C})$. Indeed, given $C\in\mathcal{C}\setminus\mathcal{C}^{*}$,
let $i\in[t]$ be minimal such that $C_{i}=C\cap E(Q_{i})$ is not
consistent. Then $C_{i}$ corresponds to one of the edge-sets in the
last four columns of Figure~\ref{fig: eq-states}. Exchanging $C_{i}$
for the other edge-set in its column yields a cycle cover $\overline{C}\in\mathcal{\mathcal{C}\setminus\mathcal{C}^{*}}$
of the same cycle format as $C$, but with opposite weight. We thus
have $\chi_{\lambda}(C)\cdot w(C)=-\chi_{\lambda}(\overline{C})\cdot w(\overline{C})$.
As the map from $C$ to $\overline{C}$ induces a fixed-point free
involution on $\mathcal{C}\setminus\mathcal{C}^{*}$, the pairs $\{C,\overline{C}\}$
indeed partition $\mathcal{C}\setminus\mathcal{C}^{*}$, and we conclude
that
\[
\imm_{\lambda}(G)=\sum_{C\in\mathcal{C}^{*}}\chi_{\lambda}(C)\cdot w(C)+\underbrace{\sum_{C\in\mathcal{C}\setminus\mathcal{C}^{*}}\chi_{\lambda}(C)\cdot w(C)}_{=0},
\]
which proves the lemma.
\end{proof}
Using this lemma, we show how to transform instances $(H,k)$ for
$\match$ into digraphs $G$ such that $\imm_{\lambda}(G)$ equals
$\match(H,k)$ up to a multiplicative constant. Our constructions
only rely on non-vanishing tetrominos and dominos; the staircase of
$\lambda$ is ``discarded'' by introducing $z(\lambda)$ padding
vertices with self-loops. The constant arising in this construction
needs some attention. To define it, consider the partition products
$\{(2^{2}),(4)\}^{\times s}$ induced by the cycle formats of active
states; we pad the partitions in this product to partitions of $n'$
with $d-2s$ dominos and $z(\lambda)$ singletons. Formally, for fixed
$\lambda\vdash n'$ and $s\leq s(\lambda)$, define 
\begin{equation}
\theta_{s}:=\{(2^{2}),(4)\}^{\times s}\times\{(2^{d(\lambda)-2s},1^{z(\lambda)})\}.\label{eq: theta_s}
\end{equation}
In the next subsection, we then establish the crucial fact that $\chi_{\lambda}(\theta_{s})\neq0$
for relevant choices of $s$.
\begin{lem}
\label{lem: skew-G-perf}Let $H$ be a graph with $n$ vertices and
$m$ edges and let $\lambda\vdash n'$ be a partition.
\begin{enumerate}
\item If $\lambda$ has skew tetromino number $s(\lambda)\geq n/2$ and
$m-n/2$ additional dominos, i.e., $d(\lambda)\geq n/2+m$, then we
can construct an $n'$-vertex graph $G$ in polynomial time such that
\begin{equation}
\imm_{\lambda}(G)=\underbrace{(-1)^{m-n/2}\cdot2^{n/2}\cdot\chi_{\lambda}(\theta_{n/2})}_{=:c_{1}}\cdot\PerfMatch(H).\label{eq: skewgraph-eq-perf}
\end{equation}
\item For any $k\leq\frac{s(\lambda)}{3}$ such that $d(\lambda)\geq m+n+2kn-5k$:
We can construct an $n'$-vertex graph $G$ in polynomial time such
that
\[
\imm_{\lambda}(G)=\underbrace{(-1)^{m+2kn-3k}\cdot(2k)!\cdot2^{3k}\cdot\chi_{\lambda}(\theta_{3k})}_{=:c_{2}}\cdot\match(H,k).
\]
\end{enumerate}
\end{lem}

\begin{proof}
For the \textbf{first part}, we define $G$ as follows:
\begin{enumerate}
\item Replace each edge $uv\in E(H)$ with a fresh copy $Q_{uv}$ of the
edge gadget $Q$. Identify $u$ and $v$ with the endpoints of $Q$.
\item Add $d'=d(\lambda)-(n/2+m)$ \emph{padding digons} and $z(\lambda)$
isolated \emph{padding vertices} with self-loops.
\end{enumerate}
By Lemma~\ref{lem: gadget-consistent} we have $\imm_{\lambda}(G)=\imm_{\lambda}^{*}(G)$,
where $\imm_{\lambda}^{*}(G)$ sums over the set $\mathcal{C}^{*}$
of consistent cycle covers. Any cycle cover $C\in\mathcal{C}^{*}$
includes all padding elements. Apart from padding elements, $C$ consists
of the active states of some gadgets and the passive states of the
remaining gadgets; let $M(C)\subseteq E(H)$ denote the set of pairs
$uv$ such that $Q_{uv}$ is active in $C$. Since active gadget states
must be vertex-disjoint, and all vertices of $G$ must be covered
by cycles in $C$, the set $M(C)$ induces a perfect matching $M(C)\in\mathcal{M}_{n/2}(H)$
in $H$. Conversely, given a perfect matching $M$ of $H$, let $\mathcal{C}_{M}^{*}\subseteq\mathcal{C}^{*}$
denote the set of consistent cycle covers with $M(C)=M$. By grouping
the terms in $\imm_{\lambda}^{*}(G)$, we obtain 
\begin{equation}
\imm_{\lambda}^{*}(G)=\sum_{M\in\mathcal{M}_{n/2}(H)}\underbrace{\sum_{C\in\mathcal{C}_{M}^{*}}\chi_{\lambda}(C)\cdot w(C)}_{=:a(M)}.\label{eq: imm-match-expanded-1}
\end{equation}
To calculate $a(M)$, we investigate the set $\mathcal{C}_{M}^{*}$:
Each cycle cover $C\in\mathcal{C}_{M}^{*}$ is obtained by 
\begin{itemize}
\item choosing an active state for each of the $n/2$ gadgets $Q_{uv}$
with $uv\in M$, each inducing the weight $w(uv)$,
\item adding the passive state (of weight $-1$) at the remaining gadgets,
for a total weight of $(-1)^{m-n/2}$, and
\item adding all padding elements, all of weight $1$.
\end{itemize}
The total weight is thus $(-1)^{m-n/2}$. As choices can only be made
at active states, the formats of cycle covers in $\mathcal{C}_{M}^{*}$
are given by the partition product 
\[
\{(2^{2}),(2^{2}),(4),(4)\}^{\times n/2}\times\{(2^{d(\lambda)-n},1^{z(\lambda)})\}.
\]
Viewed as a multiset of partitions, this product amounts to $2^{n/2}$
copies of $\theta_{n/2}$: After choosing one of the formats $\{(2^{2}),(4)\}$
for each of $n/2$ entries, we can choose the first or second copy
of this format. This shows that $a(M)=(-1)^{m-n/2}\cdot2^{n/2}\cdot\chi_{\lambda}(\theta_{n/2})$
and thus proves (\ref{eq: skewgraph-eq-perf}).

For the \textbf{second part}, we construct the graph $G$ in a similar
way, but we need to add some additional structures to account for
the fact that most vertices are unmatched in a $k$-matching for $k\ll n$.
\begin{enumerate}
\item Replace all edges of $H$ by edge gadgets.
\item For each vertex $v\in V(H)$, add a \emph{switch vertex} $s_{v}$
and connect it to $v$ with a \emph{switch digon}.
\item Add $2k$ \emph{receptor vertices}. Add a receptor edge between each
pair of receptor and switch vertex, then replace these edges by edge
gadgets.
\item Add $d'=d(\lambda)-(m+n+2kn-5k)$ isolated digons and $z(\lambda)$
isolated vertices.
\end{enumerate}
Any cycle cover $C\in\mathcal{C}^{*}$ then consists of the following
cycles:
\begin{itemize}
\item Each receptor vertex must be covered by an active edge gadget. Then
the other endpoint of that gadget is the switch vertex of some vertex
in $H$. The remaining $2k\cdot(n-1)$ edge gadgets incident with
receptor vertices are passive. There are $(2k)!$ ways of matching
the $2k$ receptor vertices to $2k$ fixed switch vertices with active
gadgets.
\item The $n-2k$ switch vertices not touched by active gadgets from receptor
vertices must be covered by switch digons.
\item By the previous item, $2k$ vertices in $H$ are left to be covered
by $k$ active edge gadgets that represent edges in $H$. As active
edge gadgets are vertex-disjoint, they induce a $k$-matching $M(C)$
in $H$, and they contribute weight $\prod_{uv\in M(C)}w(uv)$.
\item Overall, there are $m-k+2k\cdot(n-1)=m+2kn-3k$ passive edge gadgets
in $C$, each contributing weight $-1$. With padding digons and loops
in $G$, there are $d(\lambda)-6k$ digons and $z(\lambda)$ loops
in $C$.
\end{itemize}
The third item describes how $C$ induces a $k$-matching $M(C)\in\mathcal{M}_{k}(H)$.
Conversely, we can observe (as in the first case) that any $k$-matching
$M$ of $H$ induces consistent cycle covers with a total contribution
of 
\[
(-1)^{m+2kn-3k}\cdot(2k)!\cdot2^{3k}\cdot w(M)\cdot\chi_{\lambda}(\theta_{3k}).
\]
Note that the factor $(2k)!$ stems from the different ways receptor
vertices can match to switch vertices.
\end{proof}
\begin{rem*}
Why not use self-loops to handle the switch vertices of matched vertices?
It seemed to work fine when dealing with staircases in Lemma~\ref{lem: staircase-graph}?
The reason is that doing so may require us to peel less than $d(\lambda)$
dominos from $\lambda$ to accommodate these self-loops. As a toy
example for the complications that can arise this way, consider $\lambda=(3,1^{2})$
with $d(\lambda)=2$ and $\rho=(2,1^{3})$. There is only one domino
in $\rho$, and we have $\chi_{\lambda}(\rho)=0$, since this single
domino can be peeled as $\tHDom$ or $\tVDom$, leading to a cancellation.
If we peel \emph{all }$d(\lambda)$ dominos from $\lambda$, then
no such cancellations can occur, since Lemma~\ref{lem: domino-tiling-sign}
guarantees the same parity of the peeled domino tilings.
\end{rem*}

\subsection{\label{subsec: tetro-char}Analyzing the character values}

In this section, we analyze $\chi_{\lambda}(\theta_{s})$ for the
partition product $\theta_{s}$ defined in (\ref{eq: theta_s}). First,
we collect a few facts on $F$ that can be checked by simple manual
calculations. In the following, recall that a \emph{domino tiling}
of a $2t$-box skew shape $\gamma$ is a border strip tableau of format
$(2^{t})$. We call it \emph{even/odd} if its number of vertical $\tVDom$
pieces is even/odd. By Lemma~\ref{lem: domino-tiling-sign}, all
domino tilings of $\gamma$ have the same parity. Also recall Section~\ref{subsec: The-Murnaghan-Nakayama-rule}
for the required definitions related to partition products.
\begin{fact}
\label{fact: nonvan-tetro}For $F=\{(2^{2}),(4)\}$, we have $\alpha_{F}(\gamma)\neq0$
iff $\gamma$ is a non-vanishing tetromino. Furthermore, for any non-vanishing
tetromino, the sign of $\alpha_{F}(\gamma)$ is given by the parity
of its domino tilings, see Table~\ref{tab: nv-tetro}:
\[
\alpha_{F}(\gamma)=\begin{cases}
+2 & \text{\ensuremath{\gamma}}\text{ has even domino tilings},\\
-2 & \text{\ensuremath{\gamma}}\text{ has odd domino tilings.}
\end{cases}
\]
\end{fact}

\begin{proof}
Let $\Gamma_{F}$ denote the skew shapes that admit a border strip
tableau with a format from $F=\{(2^{2}),(4)\}$. We observe that every
shape in $\Gamma_{F}$ admits a domino tiling with two dominos, and
that $\Gamma_{F}$ therefore is a subset of the nine connected $4$-box
skew shapes from Table~\ref{tab: conn-shapes} and the four disconnected
two-domino skew shapes. Furthermore:
\begin{itemize}
\item We have $\alpha_{F}(\gamma)\in\{-2,2\}$ if $\gamma$ consists of
two disconnected dominos, and the sign is positive iff both dominos
have the same orientation. Indeed, $\gamma$ admits exactly two border
strip tableaux $T_{1},T_{2}$ with formats from $F$, which are domino
tilings with the same height sign.
\item We have $\alpha_{F}(\tSquare)=2$, since the square admits two border
strip tableaux with block formats from $F$. Both tableaux are domino
tilings with the same height sign.
\item The remaining connected $4$-box shapes $\gamma$ are border strips.
As such, they admit one border strip tableau of format $(2^{2})$
and $(4)$ each, implying that $\alpha_{F}(\gamma)\in\{-2,0,2\}$.
We have $\alpha_{F}(\gamma)\neq0$ iff the height of $\gamma$ and
the unique domino tiling of $\gamma$ agree in parity. (Recall that
the height $\height(\gamma)$ is the number of touched rows \emph{minus
one.}) The sign is positive iff the tiling is even.
\end{itemize}
These observations together prove the statement.
\end{proof}
\begin{table}
\begin{centering}
\begin{tabular}{c>{\centering}m{1.5cm}|>{\centering}m{1.5cm}|>{\centering}m{1.5cm}|>{\centering}m{1.5cm}>{\centering}m{1.5cm}}
positive: & \ydiagram{4} & \ydiagram{2,2} & \ydiagram{1+1,2,1} & \multicolumn{1}{>{\centering}m{1.5cm}|}{\ydiagram{3+2,2}} & \ydiagram{2+1,2+1,1,1}\tabularnewline
 & \multicolumn{1}{>{\centering}m{1.5cm}}{~} & \multicolumn{1}{>{\centering}m{1.5cm}}{} & \multicolumn{1}{>{\centering}m{1.5cm}}{} &  & \tabularnewline
 & \multicolumn{1}{>{\centering}m{1.5cm}}{~} & \multicolumn{1}{>{\centering}m{1.5cm}}{} & \multicolumn{1}{>{\centering}m{1.5cm}}{} &  & \tabularnewline
negative: & \ydiagram{3,1} & \ydiagram{2+1,3} & \ydiagram{2+2,1,1} & \ydiagram{3+1,3+1,2} & \tabularnewline
\end{tabular}
\par\end{centering}
~

\caption{\label{tab: nv-tetro}The shapes $\gamma$ with $\alpha_{F}(\gamma)\protect\neq0$,
grouped by sign.}
\end{table}

We can now prove that the relevant characters in the tetromino-based
reduction do not vanish. 
\begin{lem}
For any partition $\lambda\vdash n'$ and integer $s\leq s(\lambda)$,
we have $\chi_{\lambda}(\theta_{s})\neq0$.
\end{lem}

\begin{proof}
By Fact~\ref{fact: nonvan-tetro}, we have $\gamma\in\Gamma_{F}$
and $\alpha_{F}(\gamma)\neq0$ for a $4$-box skew shape $\gamma$
iff $\gamma$ is a non-vanishing tetromino. Let $D=\{(2)\}$ and $S=\{(1)\}$;
then $\Gamma_{D}=\{\tHDom,\tVDom\}$ and $\Gamma_{S}=\{\boxempty\}$.
Define the set of skew shape tableaux
\[
\mathcal{S}=\mathcal{S}(\lambda,\underbrace{\Gamma_{F},\ldots,\Gamma_{F}}_{s\text{ times}},\underbrace{\Gamma_{D},\ldots,\Gamma_{D}}_{d(\lambda)-2s\text{ times}},\underbrace{\Gamma_{S},\ldots,\Gamma_{S}}_{z(\lambda)\text{ times}})
\]
and let $t=d(\lambda)-s+z(\lambda)$ be the number of sets of skew
shapes in the above list. For $i\in[t]$, let $\alpha_{i}\in\{\alpha_{F},\alpha_{D},\alpha_{S}\}$
be the coefficient function for the $i$-th set in the list. By Lemma~\ref{lem: mn-rule extended},
we have
\begin{equation}
\chi_{\lambda}(\theta_{s})=\sum_{\substack{S\in\mathcal{S}\text{ with}\\
\mathrm{shapes\ }\gamma_{1}\ldots\gamma_{t}
}
}\ \prod_{i=1}^{t}\alpha_{i}(\gamma_{i}).\label{eq: MN-R proof-1}
\end{equation}
It follows that every tableau $S\in\mathcal{S}$ with non-zero weight
in the above sum peels $s$ non-vanishing tetrominos from $\lambda$,
followed by $d(\lambda)-2s$ dominos, and $z(\lambda)$ singleton
boxes. The tetrominos and dominos tile $\lambda/\mu$, where $\mu$
is the staircase of $\lambda$.

Since $s\leq s(\lambda)$, at least one tableau $S\in\mathcal{S}$
exists, and we can prove the lemma by showing that all tableaux are
counted with the same sign in (\ref{eq: MN-R proof-1}). Towards this,
note that any skew shape tableau $S\in\mathcal{S}$ can be turned
into a border strip tableau $B(S)$ of $\lambda$ that peels $d(\lambda)$
dominos and $z(\lambda)$ singleton boxes from $\lambda$: Simply
tile each tetromino $\gamma$ in $S$ arbitrarily with two dominos.
By Fact~\ref{fact: nonvan-tetro}, we know that $\alpha_{F}(\gamma)$
is positive/negative iff the tiling of $\gamma$ is even/odd, so $\prod_{i=1}^{t}\alpha_{i}(\gamma_{i})$
is positive/negative iff the domino tiling of $\lambda/\mu$ induced
by $B(S)$ is even/odd. (Singleton boxes can be ignored, as they contribute
the factor $+1$.) But by Lemma~\ref{lem: domino-tiling-sign}, all
domino tilings of $\lambda/\mu$ have the same parity. Therefore,
all terms in (\ref{eq: MN-R proof-1}) have the same sign.
\end{proof}
\begin{cor}
\label{cor: skew-coeffs}The coefficients $c_{1}$ and $c_{2}$ defined
in Lemma~\ref{lem: skew-G-perf} are both non-zero.
\end{cor}

\subsection{\label{subsec: tetro-hardness}Reductions}

As in Section~\ref{subsec: staircase-hardness}, we collect the previous
arguments to obtain a reduction from counting matchings to immanants
for partitions with large non-vanishing tetromino number.
\begin{lem}
\label{lem: finalred-tetro}The following can be achieved in polynomial
time and with polynomial-sized arithmetic circuits: 
\begin{enumerate}
\item Given an $n$-vertex graph $H$ of maximum degree $3$ and a partition
$\lambda$ with $s(\lambda)\geq n/2$ and $d(\lambda)\geq2n$, compute
a digraph $G$ and a number $c\in\mathbb{Q}$ such that $\PerfMatch(H)=c\cdot\imm_{\lambda}(G)$.
\item Given an $n$-vertex graph $H$ and $k\in\mathbb{N}$, and a partition
$\lambda$ with $s(\lambda)\geq3k$ and $d(\lambda)\geq n^{2}+n+2kn-5k$,
compute a digraph $G$ and a number $c\in\mathbb{Q}$ such that $\match(H,k)=c\cdot\imm_{\lambda}(G)$.
\end{enumerate}
\end{lem}

\begin{proof}
For the \textbf{first part}, let $H$ be a graph with $n$ vertices
and maximum degree $3$, so that $m\leq\frac{3}{2}n$. Note that $d(\lambda)\geq n+m$
by assumption. We construct a graph $G$ via Lemma~\ref{lem: skew-G-perf}
such that $\imm_{\lambda}(G)=c_{1}\cdot\PerfMatch(H)$. For the \textbf{second
part}, we can use Lemma~\ref{lem: skew-G-perf} to construct a graph
$G$ such that $\imm_{\lambda}(G)=c_{2}\cdot\match(H,k)$. We have
$c_{1},c_{2}\neq0$ by Corollary~\ref{cor: skew-coeffs}.
\end{proof}

\section{\label{sec:Completing-the-proof}Completing the proofs}

Let $\Lambda$ be a family of partitions with unbounded $b(\Lambda)$.
We compose the constructions from the preceding sections to an overall
hardness proof for $\immProb(\Lambda)$. This requires us to find
sequences of partitions within $\Lambda$ that are dense enough and
supply sufficiently many boxes to the right of the first column. In
the sub-polynomial growth regime, we also need to ensure a sufficiently
large supply of dominos; this can be achieved by having a large number
of boxes in the first column.
\begin{defn}
Given a polynomial-time computable function $g:\mathbb{N}\to\mathbb{N}$,
a family of partitions $\Lambda$ \emph{supports growth }$g$ if there
is a sequence $\Lambda'=(\lambda^{(1)},\lambda^{(2)},\ldots)$ in
$\Lambda$ such that $\lambda^{(n)}$ satisfies $b(\lambda^{(n)})\geq g(n)$
and $\left\Vert \lambda^{(n)}\right\Vert =\Theta(n)$. We also say
that $\Lambda$ supports growth $g$ \emph{via} $\Lambda'$. We say
that $\Lambda$ \emph{computationally} supports growth $g$ if $\lambda^{(n)}$
can be computed in polynomial time from $n$.
\end{defn}

We are ready to prove the main theorems of this paper. To this end,
we distinguish whether $\Lambda$ supports polynomial growth (for
a reduction from $\PerfMatch$) or only sub-polynomial growth (for
a reduction from $\match^{(g)}$ for any growth $g$ supported by
$\Lambda$). Recall that, by known algorithms~\cite{doi:10.1080/03081088508817680,DBLP:books/daglib/0025071},
we have $\immProb(\Lambda)\in\P$ and $\immProb(\Lambda)\in\VP$ for
any family $\Lambda$ with $b(\Lambda)<\infty$.
\begin{proof}[Proof of Theorem~\ref{thm: main-poly}]
The tractability part is known. For hardness, we reduce from $\PerfMatch$:
Let $H$ be an $n$-vertex bipartite graph for which we want to compute
$\PerfMatch(H)$. Let $\alpha>0$ be maximal such that $\Lambda$
supports growth $\Omega(n^{\alpha})$. We find a partition $\lambda\in\Lambda$
with $\abs{\lambda}=\Theta(n^{1/\alpha})$ such that $b(\lambda)\geq20n$.
Then Lemma~\ref{lem: mine-main} guarantees that at least one of
$s(\lambda)\geq3.5n$ or $w(\lambda)\geq\sqrt{20n}-1$ holds. In either
case, using the first cases of Lemmas~\ref{lem: finalred-staircase}
and \ref{lem: finalred-tetro}, we compute a graph $G$ on $\abs{\lambda}$
vertices and a number $c\in\mathbb{Q}$ such that $\PerfMatch(H)=c\cdot\imm_{\lambda}(G)$
holds. Overall, this yields polynomial-time and c-reductions from
$\PerfMatch$ to $\immProb(\Lambda)$, showing $\sharpP$-hardness
and $\VNP$-completeness of the latter. The lower bound under $\sharpETH$
follows since $\abs{\lambda}=\Theta(n^{1/\alpha})$.
\end{proof}
If $\Lambda$ supports only sub-polynomial growth, the proof proceeds
similarly. In this case, we can find a sequence of partitions in which
most rows have width $1$, which allows us to peel a large amount
of dominos from the left-most column.
\begin{proof}[Proof of Theorem~\ref{thm: main-param}]
If $\Lambda$ supports polynomial growth, we use Theorem~\ref{thm: main-poly}.
Otherwise, let $g\in\omega(1)$ be a growth supported by $\Lambda$.
We reduce from $\match^{(h)}$ with $h(n)=\sqrt{g(n)}/24$: Let $(H,k)$
be an instance for $\match^{(h)}$ with an $n$-vertex graph $H$
and $k\leq h(n)$. Using the growth condition on $\Lambda$ and $g\in O(n^{0.1})$,
we find a partition $\lambda\in\Lambda$ with $b(\lambda)\geq24k$
and at least $2n^{2}+b(\lambda)$ boxes in the first column, which
implies $d(\lambda)\geq n^{2}$. With the bound on $b(\lambda)$,
Lemma~\ref{lem: mine-main} yields that $s(\lambda)\geq3k$ or $w(\lambda)\geq\sqrt{24k}-1$.
In either case, using the second case of Lemmas~\ref{lem: finalred-staircase}
and \ref{lem: finalred-tetro}, we compute a graph $G$ on $\abs{\lambda}$
vertices and a number $c\in\mathbb{Q}$ such that $\match(H,k)=c\cdot\coeff{x^{t}}{\imm_{\lambda}(G)}$,
where $t=2k$ for the staircase-based reduction (Lemma~\ref{lem: finalred-staircase})
and $t=0$ for the tetromino-based reduction (Lemma~\ref{lem: finalred-tetro}). 

Note that we can compute $\coeff{x^{t}}{\imm_{\lambda}(G)}$ via polynomial
interpolation from the values $\imm_{\lambda}(G_{i})$ for $i\in\{0,\ldots,\abs{\lambda}\}$,
where $G_{i}$ is the graph obtained from $G$ by replacing each edge
of weight $x$ with an edge of weight $i$. Overall, we obtain a polynomial-time
Turing reduction from $\match^{(h)}$ to $\immProb(\Lambda)$, which
implies by Lemma~\ref{thm: match-bit-hardness} that $\immProb(\Lambda)\notin\P$
unless $\FPT=\sharpWone$. An analogous statement holds in the algebraic
setting: As polynomial interpolation amounts to solving a system of
linear equations, which can be performed with polynomial-sized circuits,
we obtain a parameterized c-reduction from the p-family $\match^{(h)}$
to the p-family $\immProb(\Lambda')$.
\end{proof}

\section{Conclusion and future work\label{sec: conclusion}}

We established a dichotomy for the complexity of immanants, concluding
a sequence of partial results obtained throughout the last four decades~\cite{doi:10.1080/03081088508817680,DBLP:journals/siamcomp/Burgisser00a,DBLP:journals/corr/cs-CC-0301024,DBLP:journals/toc/MertensM13,DBLP:conf/cie/Rugy-Altherre13}. 

Let us note that immanants are not the only way of generalizing permanents
and determinants into a family of matrix forms. Other examples include
the \emph{cover polynomials} \cite{DBLP:journals/jct/ChungG95,DBLP:journals/cc/BlaserDF12}
and the \emph{fermionants} \cite{DBLP:journals/toc/MertensM13,DBLP:journals/algorithmica/BjorklundKW19},
which are also sum-products over row-column permutations of a matrix.
Unlike the immanant however, these families do \emph{not} exhibit
gradual progressions from easy to hard cases, and they feature no
non-trivial easy cases beside the determinant.

Our result for immanants also prompts several interesting follow-up
questions.

\paragraph{Modular immanants.}

Writing $\lambda'$ for the transpose of a partition $\lambda$, it
is known that $\chi_{\lambda}$ and $\chi_{\lambda'}$ are equivalent
modulo $2$. This renders $\lambda$-immanants tractable over $\mathbb{Z}_{2}$
for partitions $\lambda$ with constantly many boxes outside of the
first column \emph{or row}. Are these the only tractable immanants
over $\mathbb{Z}_{2}$? Which immanants are tractable over $\mathbb{Z}_{p}$
for odd primes $p$? These questions lead into the tricky territory
of representation theory over fields of positive characteristic.

\paragraph{Planar graphs.}

The permanent is polynomial-time solvable for bi-adjacency matrices
of planar bipartite graphs. This is a consequence of the classical
FKT algorithm~\cite{doi:10.1080/14786436108243366}, which expresses
the number of perfect matchings in a planar graph $G$ as a determinant.
The relevant matrix is derived from the adjacency matrix of $G$ by
flipping signs according to a \emph{Pfaffian orientation} of $G$.
Can this algorithmic idea be generalized to more general immanants
on bi-adjacency matrices of planar bipartite graphs? Which immanants
remain hard on planar graphs?

\paragraph{Removing weights.}

Our proof establishes hardness for matrices with general entries from
$\mathbb{Z}$: Even though the source problems $\PerfMatch$ and $\match$
are hard for unweighted graphs, the edge-gadget and interpolation
steps introduce non-unit weights. It may however still be possible
to establish hardness for $0$-$1$ matrices, as known for the permanent.
While gadgets for simulating edge-weights are known, they change the
formats of the relevant cycle covers, and it is more difficult to
argue about character values on these formats.

\section*{Acknowledgments}

I thank Christian Engels for introducing me to the immanants and Nathan
Lindzey for introducing me to the representation theory of $S_{n}$.

\bibliographystyle{plain}
\bibliography{imm}

\end{document}